\documentclass[aps,onecolumn,superscriptaddress,nofootinbib,a4paper,longbibliography]{revtex4-2}

\usepackage{times}
\usepackage{epsfig}
\usepackage{amsfonts}
\usepackage{amsmath}
\usepackage{amsthm}
\usepackage{amssymb}					
\usepackage{amsthm}
\usepackage{dsfont}
\usepackage{bm}
\usepackage{mathtools}
\usepackage{color}
\usepackage{multirow}
\usepackage[normalem]{ulem}
\newcommand{\stkout}[1]{\ifmmode\text{\sout{\ensuremath{#1}}}\else\sout{#1}\fi}
\usepackage{latexsym}
\usepackage{mathrsfs}
\usepackage{natbib}
\usepackage{verbatim}
\usepackage[T1]{fontenc}
\usepackage{float}
\usepackage{graphicx}
\usepackage{subcaption}
\usepackage{xcolor}
\usepackage{physics}
\usepackage{soul}
\usepackage{caption}
\usepackage{subcaption}

\usepackage{centernot}
\usepackage[font=small,labelfont=bf]{caption}

\usepackage{xspace}

\newtheorem{theorem}{Theorem}

\newtheorem{corollary}[theorem]{Corollary}

\newtheorem{result}{Result} 

\newtheorem{definition}{Definition}

\newcommand{\floor}[1]{\lfloor #1 \rfloor}

\newcommand{\bracket}[3]{\langle#1|#2|#3\rangle}

\usepackage{amssymb}
\usepackage{amsthm}
\usepackage{mathtools} 
\usepackage[customcolors]{hf-tikz} 
\usetikzlibrary{patterns}
\usetikzlibrary{matrix,decorations.pathreplacing}

\pgfkeys{tikz/mymatrixenv/.style={decoration={brace},every left delimiter/.style={xshift=8pt},every right delimiter/.style={xshift=-8pt}}}
\pgfkeys{tikz/mymatrix/.style={matrix of math nodes,nodes in empty cells,left delimiter={[},right delimiter={]},inner sep=1pt,outer sep=1.5pt,column sep=2pt,row sep=2pt,nodes={minimum width=20pt,minimum height=10pt,anchor=center,inner sep=0pt,outer sep=0pt}}}
\pgfkeys{tikz/mymatrixbrace/.style={decorate,thick}}

\usepackage[colorlinks=true,linkcolor=blue,citecolor=magenta,urlcolor=blue]{hyperref}

\begin{document}
	

	\title{Simulating quantum measurements without superposition devices}

	\author{Gabriele Cobucci}\thanks{Corresponding author: gabriele.cobucci@fysik.lu.se}
	\affiliation{Physics Department and NanoLund, Lund University, Box 118, 22100 Lund, Sweden.}

	\author{Alexander Bernal}
	\affiliation{Physics Department and NanoLund, Lund University, Box 118, 22100 Lund, Sweden.}
	\affiliation{Instituto de f\'isica Te\'orica, IFT-UAM/CSIC,
		Universidad Aut\'onoma de Madrid, Cantoblanco, 28049 Madrid, Spain.}

	\author{Roope Uola}
	\affiliation{Department of Physics and Astronomy, Uppsala University, Box 516, 751 20 Uppsala, Sweden.}
	\affiliation{Nordita, KTH Royal Institute of Technology and Stockholm University,
		Hannes Alfv\'ens v\"ag 12, 10691 Stockholm, Sweden}
	
	\author{Armin Tavakoli}
	\affiliation{Physics Department and NanoLund, Lund University, Box 118, 22100 Lund, Sweden.}

\begin{abstract}
	Superposition is the core feature that sets quantum theory apart from classical physics. Here, we investigate whether sets of quantum measurements can be modelled by using  devices limited to performing jointly diagonalisable measurements.  This leads us to introduce classical measurement models, where many such devices are connected together through classical randomness and post-processing. Such classical measurement models are intermediate between the notion of commutative measurements and joint measurability. Towards understanding these models we (i) identify exact noise and loss rates at which all projective measurements admit a classical model,  (ii) propose numerical methods to construct classical models for finite sets of measurements, and (iii) show how to construct witnesses for ruling out the existence of classical models. In addition, we show that classical measurement models also have operational implications in non-disturbance tasks where sequential quantum measurements are implemented with classical side-information.  Our work provides a new approach to the role of superposition in quantum measurements.
\end{abstract}

	\date{\today}
	\maketitle

\section{Introduction} The superposition principle is at the heart of quantum theory's break with classical ideas. One possible way to view superposition is to define a pre-selected basis and consider the existence of off-diagonal contributions in that basis \cite{Streltsov2017}. A different and less restrictive view is to adopt no privileged basis and ask whether there exists some basis in which no off-diagonal contributions are necessary. Such joint diagonalisability is captured by the textbook notion of commuting observables. For projective measurements, commutation is equivalent to joint measurability (JM), by which there is a single measurement from which the others are obtained by classical post-processing \cite{Guhne2023}. However, for general measurements, which are  positive operator-valued measures (POVMs), commutation is a strictly stronger condition than JM. It motivates a natural question: to what extent can the measurements in quantum theory be simulated by classically arranging joint diagonalisable measurement devices?


Here, we introduce such simulation models, in which individually commuting measurement devices are coordinated through classical operations. Since all  the measurements implemented by any given measurement device are jointly diagonalisable, we refer to these as classical measurement models. Their goal is to reproduce the Hilbert space representation of any set of POVMs implemented by a given quantum measurement device. Operationally, one may think of classical measurement models as an independent random variable being used to select between a number of measurement devices that individually implement only commuting observables, and whose output is subjected to a conventional post-processing. Naturally, standard commuting measurements are a special case of this framework, but the possibility of stochastically calling the different devices allows our models to account for  a strictly larger class of measurements than those captured by commutation. For instance, the noisy Pauli observables $\eta \sigma_Z$ and $\eta \sigma_X$ fail to commute for any $\eta>0$, but we will see that classical measurement models fail only for $\eta> \frac{1}{\sqrt{2}}$.  While weaker than commutation, classical measurement models are stronger than JM, as illustrated in  Fig~\ref{fig_sets}a. This is because JM addresses whether a set of measurements can be reduced to a single measurement, which is necessary for a description without superposition, but does not address  whether that single measurement   needs superposition features. This makes JM natural for certain black-box input/output scenarios \cite{Wolf2009, Quintino2014, Uola2015}, but also too powerful  for capturing whether superposition devices are essential for  the measurements.

\begin{figure}[h!]
	\centering
	\includegraphics[width=1\columnwidth]{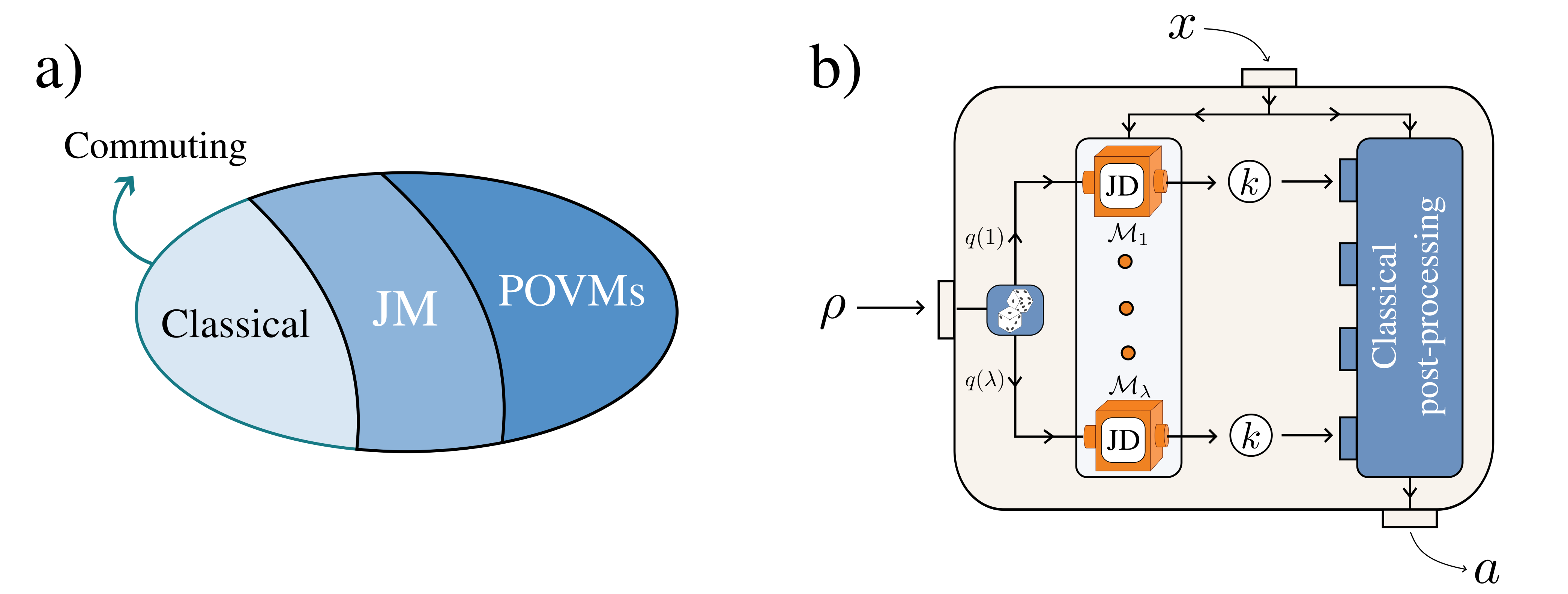}
	\caption{a) \textbf{Classifying measurement models.} Commuting observables, joint measurability and POVMs represent increasingly broad notions of measurement. We introduce classical measurement models based on classically wiring  devices capable only of performing commuting measurements. These models are intermediate between  commutation and joint measurability. b) \textbf{Classical measurement model.} A random variable, $\lambda$, is used to direct the incoming quantum state, $\rho$, to a measurement device called $\mathcal{M}_\lambda$. This device takes $x$ as input and performs the measurement $F_{k|x,\lambda}$ chosen from a set of jointly diagonalisable (JD) measurements over $x$. The outcome $k$, the device label $\lambda$ and the label of the quantum measurement $x$,  are post-processed into the output  $a$.}
	\label{fig_sets}
\end{figure}


Apart from JM, classical measurement models can be compared with other notions of measurement simulation in the quantum information literature, such as those in Refs~\cite{Guerini2017} and \cite{Oszmaniec2017}. Ref~\cite{Guerini2017} introduces a general framework for simulating quantum measurements via a classical processing of accessible measurements with constrained properties. The framework has a different motivation to ours and as we shall argue, their focus-case of outcome-limited simulators reproduces classical models in the special case of two-outcome qubit measurements. In contrast, Ref~\cite{Oszmaniec2017} asks whether a single measurement can be simulated by classically combining projective measurements. As we will see, this definition and motivation emerges as the simplest special case of our framework, namely the one in which the set of measurements to be simulated  has only a single member.	


In what follows, we first formalise  the classical measurement models. In particular, we identify their relationship to JM and use this to highlight how these models can be fully characterised through hierarchies of semidefinite programming relaxations \cite{Tavakoli2024}. We then discuss to what extent classical measurement models can account for all the projective measurements in $d$-dimensional quantum theory. We provide an exact solution to this problem both in terms of the depolarisation noise and the measurement loss needed to make a classical model possible. Subsequently, we consider more realistic quantum measurement devices, which only implement some finite set of well-selected measurements. For any finite set of POVMs, we show how to numerically search for a classical measurement model. Furthermore, should no such model exists, we  show how  it is possible to obtain proofs of this fact from witness methods tailored for the prepare-and-measure scenario, which is a common setting for certification experiments \cite{Brask2026}. Then, we go beyond the foundational interest in understanding the role of superposition in quantum measurements and consider the operational implications of classical measurement models. Specifically, we consider the well-known problem of whether a pair of measurements can be implemented in sequence without the former disturbing the latter \cite{Guhne2023}. We consider this problem in the presence of a classical hidden variable, shared between  the devices tasked with implementing the respective measurements, and we show that a classical measurement model implies the possibility of a non-disturbing implementation. Classical models therefore operationally contrast JM, as the latter turns out only to be a  necessary condition for the task. We conclude with a discussion about the relationship between the classical measurement models proposed in this work and the notion of classical state models recently proposed in Ref~\cite{Cobucci2025}.

\section{Classical measurement models}\label{section_classical_model}
Consider a set of POVMs, $\mathbf{M}\equiv \{M_{a|x}\}_{a,x}$, on a $d$-dimensional Hilbert space $\mathcal{H}$, where $a=1,\ldots,o$ labels the outcome and $x=1,\ldots,n$ labels the measurement choice. We ask whether  $\mathbf{M}$ can be simulated in a model that uses only the following resources: (i) an independent random variable, (ii) arbitrary post-processing, and (iii) measurement devices that individually are jointly diagonalisable on the Hilbert space of the system. This model is illustrated in Fig~\ref{fig_sets}b and we now proceed to formalise it.

We consider an arbitrary quantum state, $\rho$, and whether it is possible to simulate its measurement statistics for any one of the measurements appearing in the set $\mathbf{M}$.  Our classical measurement model for the implementation of $\mathbf{M}$ is as follows. We use a random variable, $\lambda$,  sampled from a probability density function $q(\lambda)$. Depending on its value, $\rho$ is directed to a measurement device named $\mathcal{M}_\lambda$. Since $\mathcal{M}_\lambda$ cannot have superposition properties, it is restricted to operate in a specific orthonormal basis of $\mathcal{H}$, independently of the input $x$. We label the measurements diagonal in the given basis by $\{F_{k|x,\lambda}\}_{k}$.  The outcome of this measurement, $k$, the device's label $\lambda$ and  the label, $x$, of the specific quantum measurement in $\mathbf{M}$ that we aim to implement, are then fed into a post-processing unit.  Thus, the post-processing is represented by some conditional distribution $p(a|x,k,\lambda)$, where $a$ is the final outcome of the simulation. Putting it all together, we can write the classical measurement model as 

\begin{equation}\label{Clmodel_predef}
M_{a|x} = \int d\lambda \, q(\lambda) \sum_{k} p(a|x,k,\lambda)F_{k|x,\lambda},
\end{equation}
where $q(\lambda)$ is a probability density function and  $[F_{k|x,\lambda},F_{k'|x',\lambda}]=0$ $\forall k,k',x,x',\lambda$.  A useful observation is that the post-processing can be dropped without loss of generality, because for every given $\lambda$ the set $\{\sum_k p(a|x,k,\lambda)F_{k|x,\lambda}\}_{a,x}$ is itself a collection of jointly diagonalisable POVMs. Conversely, every set of jointly diagonalisable POVMs can be written in the same form. We adopt this simplification directly in our definition of a classical measurement model.

\begin{definition}[Classical measurement model]\label{Def1}
	The quantum measurements $\{M_{a|x}\}_{a,x}$ are said to admit a classical  model if they can be decomposed on the form
\begin{align}\nonumber\label{Clmodel_def}
M_{a|x} =& \int d\lambda \, q(\lambda) F_{a|x,\lambda},\\\nonumber
\text{where} &\quad [F_{a|x,\lambda},F_{a'|x',\lambda}]=0 \quad \forall a,a',x,x',\lambda\\
\text{and} &\quad \{F_{a|x,\lambda}\}_a \quad \text{is a POVM } \forall x,\lambda.
\end{align}
\end{definition}

{\color{red}
}

A crucial element in this definition is that $\lambda$ is independent of $x$. If this was not the case, $\lambda$ could use the knowledge of $x$ to direct $\rho$ to the most appropriate measurement device. For instance, this would have allowed any set of projective measurements $\mathbf{M}$ to have a classical model. From a physical point of view, the independence of $\lambda$ is motivated by the idea that it is provided to the simulation by an external party or process, who has no knowledge of the free choice of $x$ made by the experimenter. This view of $q(\lambda)$ as a pre-programming, made before the experimenter makes their free choice of $x$, is analogous to the so-called  freedom of choice assumption in Bell inequality experiments \cite{Bell1977}. Furthermore, note that  the model in Eq~\eqref{Clmodel_def} may use an (uncountably) infinite number of devices $\mathcal{M}_\lambda$, thereby motivating the integral in the definition.

Definition~\ref{Def1} includes two special cases where it reduces to notions from previous literature. Firstly, for $n=1$ (a single measurement) it becomes equivalent to the notion of projective simulability introduced in \cite{Oszmaniec2017}. In this case, the motivations for the respective definitions coincide and we will soon exploit this to repurpose results from the literature to give a characterisation of classical measurement models. Secondly, it reproduces two special instances of the simulation models in \cite{Guerini2017}. The framework introduced there focuses on simulating a set of POVMs $\mathbf{M}$ using only measurements (called simulators) that satisfy certain constraints (in particular, their work specialises on limiting the number of outcomes for each simulator). The general protocol consists of a different pre-processing rule to choose a simulator for each POVM in the set and a classical post-processing to generate a specific output. By choosing a pre-processing rule that is independent of the POVMs in the set and requiring the simulators to be projective, their general framework reproduces classical models, although this case has not been explicitly studied in that context. Similarly, for $2$-outcome qubit measurements and independent pre-processing rule, their focus-case of outcome-limited simulators reproduces classical models. This can be viewed as a ``qubit accident'', since for qubit systems, extremal dichotomic measurements happen to coincide with projective measurements \cite{DAriano2005}, and so any generic simulator can serve as the resource required for our classical models.


Furthermore, the classical measurement models in Definition~\ref{Def1} can be seen as a continuation of a research line aimed at developing a general framework for classical simulation of operations in Hilbert space. The first step in this program was taken with the introduction of the classical models for simulating quantum states in Ref~\cite{Cobucci2025}. Since the core of the simulation framework is the same for both states and measurements, there is a clear parallelism between the two cases. However, the two classes of models show several differences, as we extensively discuss in Section \ref{Discussion} and in the Supplementary Material. In general, the simulation capability of the measurement models is higher than for states. Indeed, if a classical simulation exists for a set of states $\mathcal{E} = \lbrace \rho_x \rbrace_{x}$, then a classical simulation also exists for the set of binary measurements $\mathbf{M} = \lbrace M_{0|x} = \rho_x, M_{1|x} = \openone - \rho_x \rbrace_{x}$, but the converse is not true.
	 

\vspace{3mm}



Let us now make a few observations about the properties of classical measurement models.
\begin{enumerate}
	\item The space of all sets of quantum measurements that admit a classical model is closed and convex.
	
	\item  If $\mathbf{M}$ admits a classical model then it is also JM. By definition, $\mathbf{M}$ is JM if there exists a parent POVM $\{G_\mu\}_\mu$, such that $M_{a|x} = \sum_{\mu}p(a|x,\mu) G_{\mu}$ \cite{Guhne2023}. In particular, classical models are equivalent to a special case of JM where the parent POVM is restricted to admit a projective simulation \cite{Oszmaniec2017}. To show this, consider the alternative form of classical models in Definition \ref{Def1}
	\begin{equation}\label{ClmodelAlternative}
		M_{a|x}=\int d\lambda \, q(\lambda) \, \sum_{k=1}^{d} p(a|x,k,\lambda) E_{k|\lambda},
	\end{equation}
	where, since $\{F_{a|x,\lambda}\}$ commutes over $a$ and $x$, we expressed $F_{a|x,\lambda}=\sum_{k} p(a|x,k,\lambda)E_{k|\lambda}$, for some rank-one projective measurement $\{E_{k|\lambda}\}_k$. We can recover Eq~\eqref{ClmodelAlternative} from a JM model selecting a projective simulable parent POVM $G_{\mu}=\int d\lambda \, q(\lambda) \sum_{k=1}^d p(\mu| \lambda, k) E_{k|\lambda}$.

	\item If $\mathbf{M}$  admits a classical model and contains an extremal POVM, then (i) that POVM is projective, and (ii) it commutes pairwise with every other POVM in $\mathbf{M}$. We first show (i). 
We can decompose the probability distribution from Eq~\eqref{ClmodelAlternative} in terms of deterministic responses by writing $p(a|x,k,\lambda)=\sum_\gamma p(\gamma|\lambda)D_\gamma(a|x,k)$, where $\gamma$ indexes all deterministic distributions, $D_\gamma\in\{0,1\}$, mapping $(x,k)$ to $a$. Define the  measurements $Z_{a|x,\gamma,\lambda}=\sum_k D_\gamma(a|x,k) E_{k|\lambda}$ and observe that they are projective by construction. Eq~\eqref{ClmodelAlternative} now reads  $M_{a|x}=\int d\lambda q(\lambda) \sum_\gamma p(\gamma|\lambda) Z_{a|x,\gamma,\lambda}$. Let  $x=x^*$ label the extremal element of $\mathbf{M}$. Extremality, by definition, means that there is no non-trivial mixture realising $M_{a|x^*}$. This implies $Z_{a|x^*,\gamma,\lambda}=Z_{a|x^*}$ and therefore $M_{a|x^*}=Z_{a|x^*}$, which proves (i). Part (ii) follows directly from (i) and point 2, since a projective POVM is JM with some other POVM if and only if they commute.

	\item From points 2 and 3, it follows that if $\mathbf{M}$ admits a classical model and contains only extremal measurements, then all members of $\mathbf{M}$ are projective and $q(\lambda)$ in Eq~\eqref{Clmodel_def} can be taken deterministic without loss of generality. This means that,  for extremal measurements, classical measurement models are equivalent to pairwise commutativity.

\end{enumerate}
The second observation has some noteworthy consequences. For qubits and qutrits, projective simulability can be fully characterised as a single semidefinite program (SDP) \cite{Oszmaniec2017, Khandelwal2025, cobucci2025MaxNP}. For spaces with higher dimension, it can be characterised through a converging hierarchy of SDPs \cite{Brinster2025}. Interpreting classical measurement models as JM with a projective simulable parent POVM means that these models too can be characterised by these SDP techniques.  This is mainly of conceptual interest, as these methods are not practically viable for the job. The latter method scales rapidly in the relaxation level while seeming to converge slowly. In addition, it becomes far more costly when taking into account the exponential-in-$n$ scaling of the number of outcomes necessary in the parent POVM. The former method does not need relaxation levels, but its complexity still rapidly deteriorates due to the same exponential-in-$n$ scaling.

Let us conclude the introduction of the classical measurement models with a simple example of how they can be used to simulate non-commuting measurements. For this, we consider an arbitrary pair of non-commuting bases, $\{\ket{f_a}\}_{a=1}^{d}$ and $\{\ket{h_a}\}_{a=1}^{d}$. We apply to both bases the depolarisation map  $\Phi_v(X)=v X+(1-v)\Tr\left(X\right)\frac{\mathds{1}}{d}$, where $v\in[0,1]$ is the visibility. Thus, $M_{a|1}=\Phi_v(\ketbra{f_a})$ and $M_{a|2}=\Phi_v(\ketbra{h_a})$. While the two measurements  commute only for $v=0$, there exists a  classical model whenever $v\leq \frac{1}{2}$, independently of the dimension $d$. To show that, we use only two measurement devices: $\mathcal{M}_{1}$ measures in $\{\ket{f_a}\}_a$  and $\mathcal{M}_{2}$ measures in $\{\ket{h_a}\}_a$. We call them with equal probability, $q(1) = q(2) = \frac{1}{2}$, and select the post-processing as  $a=k$ if $x=\lambda$ while $a$  is random if $x\neq \lambda$. Inserted into Eq~\eqref{ClmodelAlternative}, this gives the set $\mathbf{M}$ when $v=\frac{1}{2}$. This model  shows that many non-commuting observables admit a classical model in arbitrary large dimension. However, the model is usually far from optimal, i.e.~even larger values of $v$ are typically possible with a classical description of $\mathbf{M}$. We will return to this matter in section~\ref{subsec2} and present classical models for visibilities  higher than $v=\frac{1}{2}$. 

\section{Characterising classical measurement models}
In this section we present results on  classical models for simulating  quantum measurements. We begin in section \ref{subsec1} with discussing classical simulation of all projective measurements in quantum theory under the influence of  noise and/or loss. This provides insight to the ultimate capabilities of classical measurement models. We then turn  to the more practical setting in which a quantum device is capable of implementing any measurements in a finite set. For this scenario, we present in section~\ref{subsec2}  a numerical method for explicitly constructing classical measurement models. Finally, in section~\ref{subsec3} we address the main practical question of how to certify that a quantum measurement device eludes a classical model, by showing how to design witness tests for this purpose.
	

\subsection{Classical models for all projective measurements under noise and loss}\label{subsec1}
A conceptually appealing question is to consider the capability of  classical models to simulate  all projective  measurements in quantum theory. Of course, a perfect simulation is impossible. Therefore, we set out to determine the critical amount of decoherence that quantum measurements must be subjected to in order to enable a classical simulation. The answer to such a question evidently depends on the specific form of decoherence. Here, we will consider two natural cases, namely when the measurements are subjected to isotropic noise and when they are subject to particle losses.

The set of all projective quantum measurements under the influence of isotropic noise corresponds to choosing  $\mathbf{M}^\text{noise}_{v}=\{\Phi_v\left(M_{a}\right)\}$, where $\Phi_v$ implements a noisy version of the target measurement and $M$  runs over  all the (uncountably many) bases of $d$-dimensional Hilbert space. Our goal is to determine a necessary and sufficient condition on $v\in[0,1]$ for the existence of a classical measurement model. The set of all bases can equivalently be associated with the set of unitaries $U\in \text{SU}(d)$, so we can  label the measurements by $M_{a|U}$, where $U$ serves as the measurement index. Applying the depolarisation map therefore gives us the target set  $\mathbf{M}^\text{noise}_v= \{vM_{a|U} + \frac{1-v}{d}\openone\}_{U\in \text{SU}(d)}$. Next, we consider also the situation of having losses in the already noisy measurement. This is described by choosing $\mathbf{M}^\text{loss}_{\eta,v}=\{L_\eta\left(\Phi_v\left(M_{a|U}\right)\right)\}_{U\in \text{SU}(d)}$, where $L_\eta$ implements the target measurement with efficiency $\eta$, while assigning the failed outcome $\varnothing$ otherwise. Specifically, the action is
\begin{equation}
	L_\eta (M_a) = \begin{cases}
		& \eta M_a \!\!\quad\quad\quad \text{if} \quad a=1,\ldots,d\\
		& (1-\eta)\mathds{1}\quad \text{if} \quad a=\varnothing
	\end{cases}.
\end{equation}
Our goal is to determine the necessary and sufficient condition of $\eta$ as a function of $v$ for the existence of a classical measurement model. The next result provides the noise and loss thresholds at which all projective measurements admit a classical model.

\begin{result}[Simulating all projective measurements]\label{res2}
The sets $\mathbf{M}^\text{noise}_v$ and $\mathbf{M}^\text{loss}_{\eta,v}$ respectively admit a classical measurement model if and only if $v\leq v^*$ and $\eta(v)\leq\eta^*(v)$, where
\begin{align}\label{Classicality_all_meas}
	& v^*=\frac{H_{d}-1}{d-1}\\ \label{Classicality_all_meas2}
	& \eta^*=d(1-v)^{d-1}.
\end{align}
Here, $H_{n} = \sum_{k=1}^{n}\frac{1}{k}$ is the Harmonic number. The formula for the loss is valid when $v>1/2$.
\end{result}
\begin{proof}
Classical measurement models are a subset of JM. Therefore, it follows immediately that $v^*$ can be no larger than the right-hand-side of Eq~\eqref{Classicality_all_meas}, because this bound has been derived for the joint measurability of all projective measurements under isotropic noise \cite{Uola2014}. To prove that the right-hand-side of Eq~\eqref{Classicality_all_meas} also is a lower bound on $v^*$, we construct in Supplementary Material an explicit classical simulation model drawing on the ideas introduced in \cite{Werner89,Jones07}.

To show the critical efficiency we again use that the right-hand-side of Eq~\eqref{Classicality_all_meas2} is known to be exact for the joint measurability of $\mathbf{M}^\text{loss}_{\eta,v}$ when $v > 1/2$ \cite{Sekatski2024}. Then, we mimic the associated JM-model to obtain a classical measurement model that saturates the bound. Notably, this argument works also when $v\leq 1/2$, although in this case it leads to different formulas of $(v,\eta)$; we refer to the Supplementary material for those expressions and to \cite{Sekatski2024} for a complete derivation.
\end{proof}

It is interesting to note that in the limit of large $d$, the critical visibility in Eq~\eqref{Classicality_all_meas} decays as $v \sim \frac{\log d}{d}$. This shows that it becomes rapidly expensive, i.e. it requires a larger amount of noise, for a classical model to account for all projective measurements in high-dimensional quantum theory. In particular, for qubits ($d=2$) we have $v=\frac{1}{2}$. Using the quantum steering findings of Refs~\cite{Yujie2024, Renner2025}, we can generalise this result to all qubit non-projective measurements.

\begin{corollary}\label{all_qubit_simu}
Consider the set of all qubit POVMs when passed through the depolarisation map $\Phi_v$. It admits a classical measurement model if and only if 
\begin{equation}
	v\leq \frac{1}{2}.
\end{equation}
\end{corollary}
\begin{proof}
	Classical models are equivalent to JM with projectively simulable parent POVM. Since an explicit projective parent POVM for the set of all qubit POVMs for $v\leq \frac{1}{2}$ is provided in \cite{Renner2025}, the result is direct.
\end{proof}
It still remains an open problem whether the criterion in Eq~\eqref{Classicality_all_meas} holds for general POVMs when $d>2$.

\subsection{Search for classical measurement models}\label{subsec2}
Real-life quantum devices do not implement all projective measurements allowed in quantum theory, but instead only the few measurements relevant for a specific quantum technology purpose. Therefore, we now  consider the more realistic situation in which $\mathbf{M}$ is a finite set of POVMs. To this end, we provide a numerical method  to search for classical measurement models for arbitrary finite sets of target measurements $\{M_{a\lvert x}\}$. We adopt the depolarisation visibility as a quantifier of their classical simulability. That is, we search for the largest $v$ such that  $\{\Phi_v(M_{a\lvert x})\}$ admits a classical model. The introduction of the parameter $v$ is largely a matter of convenience and quantification. If a model is found for $v\geq 1$ then we conclude that $\mathbf{M}$ admits a classical measurement model. If only a smaller value of $v$ is found, we get an estimate of how well we are able to simulate $\mathbf{M}$.

To construct the simulation, let the basis in which device $\mathcal{M}_{\lambda}$ operates be $E_{k|\lambda} = U_{\lambda}\ketbra{k}U_{\lambda}^{\dagger}$, for some unitary $U_\lambda$, and define $\tilde{q}(a,\lambda|x,k) \equiv q(\lambda)p(a|x,k,\lambda)$. For a given choice of $\lbrace U_{\lambda} \rbrace_{\lambda}$, the best classical model becomes a  linear program,
\begin{equation}\label{LP}
	\begin{aligned}
		\max_{v,q,\tilde{q}} & \quad v\\
		\text{s.t.} &\quad \Phi_v(M_{a\lvert x}) =  \sum_{\lambda} \sum_{k=1}^{d} \tilde{q}(a,\lambda|x,k) E_{k\lvert \lambda}, \quad \forall a,x,\\
		& \quad \sum_{a}\tilde{q}(a,\lambda|x,k) = q(\lambda), \quad \forall \lambda,x,k,\\
		& \quad \sum_{\lambda} q(\lambda) = 1, \qquad \tilde{q}(a,\lambda|x,k) \geq 0, \quad \forall a,\lambda,x,k.
	\end{aligned}
\end{equation}
Our implementation of this program is available at  \cite{github-code_simu-witness}. In this program, we have the freedom to select the set of unitaries. Since it is not straightforward how to best make this selection, we use a relatively large number of random unitaries.
\begin{table}[h!]
	\begin{tabular}{|c|c|c|c|c|c|c|}
		\hline
		$\,$ d $\,$ & $\,$ Set $\,$ & $\, N_{\lambda} \,$ & $\,$ Numerical $\,$ & $\,$ Analytical $\,$ & $\,$ Witness $\,$ \\
		\hline
		2 & $\mathbf{M}_1$ & 20000 & 0.7605 & 0.5$^{\bullet}$ & 0.7729\\
		\hline
		2 & $\mathbf{M}_2$ & 10000 & 0.7071 & 0.5 & 0.7071\\
		\hline
		2 & $\mathbf{M}_3$ & 10000 & 0.5774 & 0.5$^\dagger$ & 0.5774\\
		\hline
		3 & $\mathbf{M}_2$ & 20000 & 0.6457 & 0.5 & 0.6667* \\
		\hline
		7 & $\mathbf{M}_4$ & 3000 & 0.2705 & 0.2655$^\dagger$ & -\\
		\hline
		
	\end{tabular}
	\caption{Threshold visibilities for classical model of different sets of POVMs. $\mathbf{M}_1$ is the set of the five qubit SIC-POVMs obtained from the compound of five tetrahedra \cite{five_tetrahedra}. $\mathbf{M}_2$ and $\mathbf{M}_3$ are the sets of two and three MUBs, respectively, in dimension $d$. $\mathbf{M}_4$ is the set of all MUBs in dimension $d=7$. The results in `Numerical' are obtained from \eqref{LP}, where $N_{\lambda}$ is the number of unitaries used to perform the simulation. They are compared with the lower bounds, `Analytical', obtained from the analytical models, i.e. that for a pair of non-commuting bases introduced at the end of Section \ref{section_classical_model}, the bound in Corollary \ref{all_qubit_simu} (denoted with $\bullet$) and the one in \eqref{Classicality_all_meas} (denoted with $\dagger$). The last column shows the upper bounds derived from the state-discrimination witness \eqref{state_discrimination_witness} (the * denotes the value computed from the SDP \eqref{SDP_witness}). The "-" indicates the cases in which our desktop could not evaluate \eqref{SDP_witness}.}\label{TabEndMatter}
\end{table}

In Table~\ref{TabEndMatter} we provide several examples to showcase the practicality of the method. For instance, we can showcase its potential precision by considering the sets $\mathbf{M}_2$ and $\mathbf{M}_3$ of two and three qubit MUBs, respectively. Running the optimisation on a standard machine and selecting $N_{\lambda} = 10000$ random unitaries returns us visibility bounds which are very  nearly identical  to the exact analytical values. The latters are known from JM to be $v=\frac{1}{\sqrt{2}}$ \cite{Busch1986} and $v=\frac{1}{\sqrt{3}}$ \cite{Brougham2007} respectively, and the fact that they are exact also for classical measurement models follows from their respective parent POVMs being projective simulable. Simulations for more complex sets of measurements are possible as well. For instance, in the first row of Table~\ref{TabEndMatter} we have simulated five four-outcome qubit measurements, each of which  a symmetric informationally complete POVM, which corresponds to a regular tetrahedron on the Bloch sphere. The five tetrahedra are oriented to each other so that they form a so-called compound of five tetrahedra ($\mathbf{M}_1$) \cite{five_tetrahedra}. Also, in the fourth and fifth row of  Table~\ref{TabEndMatter}, we present results for sets of high-dimensional measurements.  For instance, using up to $N_{\lambda} = 20000$ random unitaries, we can outperform the analytical model introduced at the end of Section \ref{section_classical_model} for pairs of non-commuting bases in dimension $d=3$. Similarly, let $\mathbf{M}_4$ be the set of all MUBs in dimension $d=7$. We can exceed the worst-case bound in Eq~\eqref{Classicality_all_meas} already by sampling only $N_{\lambda} = 3000$ unitaries. While our program returns only lower bounds on the critical $v$ for a classical simulation, we will in the next section show that our results cannot be far from optimal. This follows because of the small gaps with upper bounds on the critical visibility, which we also display in Table~\ref{TabEndMatter}.


\subsection{Ruling out classical measurement models}\label{subsec3}
We now turn to the complementary question of showing that a given set of measurements  $\mathbf{M} = \lbrace M_{a|x} \rbrace_{a,x}$ admits no classical model.  Point 3 in the list of properties in section~\ref{section_classical_model} already points out that this can be addressed using SDP relaxations,  exploiting the observation that classical measurement models can be cast as JM problems with a projective simulable parent POVM. However,  because that approach is neither computationally nor experimentally practical, we here adopt a different one based on witness tests. These tests take place in the prepare-and-measure scenario, which is a standard experimental setting for certifying quantum devices \cite{Brask2026}.
	
Consider that we have a preparation device that can generate any one of the states  $\{\rho_1,\ldots,\rho_m\}$. Letting $z$ label our choice of state, we use $\rho_z$ to probe our measurement device for its $x$'th setting. By Born's rule, the resulting probability distribution is $p(a|z,x) = \tr(\rho_z M_{a|x})$. We seek to determine inequalities of the form 
\begin{equation}\label{Witness_form}
	W(\mathbf{M}) = \sum_{a,z,x} c_{azx} \tr(\rho_z M_{a|x}) \leq \beta,
\end{equation}
where $c_{azx}$ are some real coefficients and $\beta$ is a bound satisfied by all classical measurement models. Then, a violation of the inequality falsifies classical models. Notably, by identifying a good witness, this approach requires the use of only a few states to rule out classicality, which is far simpler, from a practical point of view, than using a tomographically complete set of states to reconstruct all the POVMs and check whether they admit a classical model.
	
We now show how to determine $\beta$ for any choice of $\mathbf{M}$, $\lbrace \rho_{1},\ldots,\rho_{m} \rbrace$ and $\{c_{azx}\}$. Assume that $\mathbf{M}$ admits a classical model as in \eqref{ClmodelAlternative}. Due to the linearity of $W$, its optimal value is reachable for a deterministic choice of $\lambda$, i.e.~we need only to consider a single classical device in the simulation which performs a basis measurement $\lbrace E_{k}\rbrace_{k}$.  Hence, we have
$W=\sum_{k,a,x} p(a|x,k)\tr\left(E_k \mathcal{O}_{ax}\right)$, where we defined  $\mathcal{O}_{ax}=\sum_{z} c_{azx} \rho_z$. Next, we write the post-processing as a convex combination of deterministic distributions $p(a|x,k) = \sum_{\gamma} q(\gamma) D_{\gamma}(a|x,k)$, where $D_{\gamma}(a|x,k) \equiv \delta_{a,\gamma(x,k)}$ \cite{TaylorMarkov}. Selecting the deterministic strategy for which the maximum is achieved, we get
\begin{equation}\label{maxstep}
		\beta = \max_{\gamma} \max_{\lbrace E_{k} \rbrace} \sum_{k=1}^{d} \tr(\sum_{a,x}D_{\gamma}(a|x,k) \mathcal{O}_{ax} E_{k}).
\end{equation}
The remaining task is to address the maximisation over $\{E_k\}$. For qubit measurements, we provide an exact analytical solution in the Supplementary Material. This leads us to the following result.

	
\begin{result}[Qubit measurement witnesses]
Let $\mathbf{M}$ be a set of classically simulable qubit POVMs. For any choice of qubit states $\{\rho_z\}_z$ and real coefficients $c_{azx}$, the witness quantity $W\left(\mathbf{M}\right)$ defined in Eq~\eqref{Witness_form} satisfies
	\begin{equation}\label{beta_witness_2}
	\beta= \max_{\gamma}\ \tr(\mathcal{O}_{+}^{(\gamma)})+\sqrt{\tr(\mathcal{O}_{-}^{(\gamma)})^2-4\det(\mathcal{O}_{-}^{(\gamma)})},
\end{equation}
where $\mathcal{O}_{\pm}^{(\gamma)}=\frac{1}{2}\sum_{a,x}\left(D_{\gamma}(a|x,1)\pm D_{\gamma}(a|x,2)\right)\mathcal{O}_{ax}$. This bound is always tight.
\end{result}
Notice that evaluating the right-hand-side of Eq~\eqref{beta_witness_2} is straightforward, since the range of $\gamma$ is finite. 

For systems with dimension $d > 2$,  we do not have an analytical solution. Therefore, instead consider an SDP  relaxation of Eq~\eqref{maxstep}. Specifically, we relax the rank-one projective measurements $\lbrace E_{k} \rbrace_{k}$ to unit trace  measurements $\lbrace N_{k} \rbrace_{k}$ and compute the SDP
	\begin{equation}\label{SDP_witness}
		\begin{aligned}
			g_{\gamma} \equiv \max_{\lbrace N_{k} \rbrace} & \quad \sum_{k=1}^{d} \tr(\sum_{a,x}D_{\gamma}(a|x,k) \mathcal{O}_{ax} N_{k})\\
			\text{s.t.} &\quad N_{k} \succeq 0, \quad \tr(N_{k}) = 1, \quad \forall k,\\
			& \quad \sum_{k=1}^{d} N_{k} = \openone,
		\end{aligned}
	\end{equation}
for each deterministic strategy $\gamma$. The bound on $\beta$ is obtained by selecting the largest value of $g_{\gamma}$ over all possible $\gamma$. Our implementation of the procedure is available at \cite{github-code_simu-witness}. We highlight that this relaxation is exact whenever the Hilbert space dimension is either $d=2$ or $d=3$, because the extreme points of the optimisation domain are precisely rank-one projective measurements for $d=2,3$ \cite{Oszmaniec2017}. It is useful to note that we can establish a simple fully analytical upper bound on $g_\gamma$, by ignoring the completeness constraint in Eq~\eqref{SDP_witness}. This decouples the $E_k$ variables and leads to
\begin{equation}\label{analytical_witness}
\beta \leq \max_\gamma	g_\gamma \leq d \max_\mu  \lambda_\text{max}\left(\sum_{a,x} D_\mu(a|x) \mathcal{O}_{ax}\right),
\end{equation} 
where $\mu$ now only runs over all maps from $x$ to $a$. This bound is naturally less exact than that obtained from Eq~\eqref{SDP_witness}, but its evaluation is much more scalable.

A central question is how to best select the witness in Eq~\eqref{Witness_form}. This is expected to be a hard problem in general, but for natural sets of target measurements we now show that simple intuition from state discrimination is sufficient to obtain good or even optimal witnesses. To illustrate the usefulness of the witness method, label the states as $\rho_z = \rho_{a,x}$, where $z=(a,x)$ for $a=1,\dots,o$ and $x=1,\dots,n$. Then, consider the following set of qubit POVMs. The first one, $\mathbf{M}_1$, is a set of five symmetric informationally complete (SIC) POVMs (where the Bloch vectors of the SICs correspond to the vertices of the compound of five tetrahedra \cite{five_tetrahedra}). The second and the third sets, $\mathbf{M}_2$ and $\mathbf{M}_3$, consist of two and three qubit MUBs \cite{Durt2010}, respectively. Select the states $\rho_{a,x}^{(i)}$ to be the normalised rank-1 operators associated with the measurements in each set $\mathbf{M}_i$ for $i \in \lbrace 1,2,3 \rbrace$ and choose $c_{azx}=c_{axby}=\delta_{a,b}\delta_{x,y}$. We can construct three different witnesses as state-discrimination protocols
\begin{equation}\label{state_discrimination_witness}
		W_{\text{st.-discr.}}^{(i)}(\mathbf{M})=\sum_{a,x}\tr(\rho_{a,x}^{(i)}M_{a|x}).
\end{equation}
The maximum value of the witness is achieved by choosing $\mathbf{M}$ to be exactly the measurement sets $\mathbf{M}_{i}$ for $i \in \lbrace 1,2,3 \rbrace$, i.e. $W_{\text{st.-discr.}}^{(1)}(\mathbf{M}_1)=10,\,W_{\text{st.-discr.}}^{(2)}(\mathbf{M}_2)=4$ and $W_{\text{st.-discr.}}^{(3)}(\mathbf{M}_3)=6$. The bounds satisfied by all measurement sets which admit a classical model can be computed from \eqref{beta_witness_2} to be 
\begin{equation}\label{Witness_Analyt_Bound}
		\begin{aligned}
			&\beta_1= 5 + \frac{\sqrt{25 + 14 \sqrt{2}}}{\sqrt{3}}\simeq 8.864,\\
			&\beta_2= 2+\sqrt{2},\quad  \beta_3= 3+\sqrt{3}.
		\end{aligned}
\end{equation}
The visibility below which each noisy measurement sets $\Phi_v\left(\mathbf{M}_i\right)$ no longer violate the respective inequality is
\begin{equation}
		\begin{aligned}
			&v_1= \frac{\sqrt{25 + 14 \sqrt{2}}}{5\sqrt{3}}\simeq 0.7729,\\
			&v_2=\frac{1}{\sqrt{2}}\simeq 0.7071,\quad v_3=\frac{1}{\sqrt{3}}\simeq 0.5774.
		\end{aligned}
\end{equation}
For the latter two cases, this is the necessary and sufficient condition, as they coincide with the result obtained in section~\ref{subsec2} by the explicit construction of a classical model and also match the known JM thresholds. For the former case, the result cannot be far from optimal, as indicated by the numerical lower bound appearing in Table~\ref{TabEndMatter}, which shows that a classical simulation is possible up to $v \leq 0.7605$.
		
Let us also illustrate examples of witnesses where we make use of the  bounds obtained from Eqs~\eqref{SDP_witness} and \eqref{analytical_witness}. Firstly, considering the same type of state discrimination witness as above, but now targeting the set of two mutually unbiased qutrit bases. We have used the SDPs in Eq~\eqref{SDP_witness} to prove the failure of classical simulation for $v>2/3$. Since the bound is exact for qutrit POVMs, the small discrepancy between the upper and lower bound (see row four in Table~\ref{TabEndMatter}) is attributed either to a somewhat suboptimal choice of witness, or a somewhat suboptimal lower bound obtained from the numerical search method in section~\ref{subsec3}. The latter reason is plausible. Secondly, consider that we have a complete set of $N$ MUBs in dimension $d$ (should it exist) and use Eq~\eqref{analytical_witness} to bound the state discrimination witness for this scenario. Specifically, we take $a=1,\ldots,d$ and $x=1,\ldots,N$ and select $\rho_{a,x}$ as the $a$'th element of the $x$'th MUB. The witness bound for classical models becomes
\begin{align}\nonumber\label{calculation_steps}
	W_{\text{st.-discr.}}&=\sum_{a=1}^d \sum_{x=1}^N\tr(\rho_{a,x}M_{a|x}) \leq d \max_\mu \lambda_\text{max}\left(\sum_{a,x}D_\mu(a|x) \rho_{a,x}\right)\\
	& = d \max_\mu \max_\phi \sum_{a,x}D_\mu(a|x)  \bracket{\phi}{\rho_{a,x}}{\phi}\leq d  \max_\phi \sum_{x} t_x(\phi) \leq N+(d-1)\sqrt{N},
\end{align}
where in the first inequality we used the bound \eqref{analytical_witness} and in the next step expressed the maximal eigenvalue calculation as an optimisation over the ray $\ket{\phi}$. In the penultimate step, we have defined $t_x(\phi)= \max_a \bracket{\phi}{\rho_{a,x}}{\phi}$. The final step is an uncertainty relation for MUBs, valid for every choice of $\ket{\phi}$, originally proven in Refs~\cite{Mandayam2010, Rastegin2013}. Choosing $N$ MUB measurements, the bound is violated in quantum theory for visibilities $v>\frac{1}{\sqrt{N}}$. This bound is insensitive to the dimension, but it provides relevant scaling when considering many measurements. A better bound can be obtained for the special case of $N=2$, because one can avoid employing the second inequality in \eqref{calculation_steps}: $\sum_{a,x}D_\mu(a|x) \rho_{a,x}$ becomes a sum of two mutually unbiased projectors. Labelling any such pair of projectors by $A$ and $B$, it follows that  $\lambda_\text{max}(A+B)=1+\frac{1}{\sqrt{d}}$ \cite{proof_MUB} and hence $W_{\text{st.-discr.}}\leq d+\sqrt{d}$. This is violated for visibilities $v>\frac{1}{2}+\frac{1}{2+2\sqrt{d}}$ which improves with larger dimension.

\section{Operational interpretation in non-disturbance tasks} 
Our motivation for introducing and analysing classical measurement models has so far rested on the conceptual interest in understanding the role of superposition in quantum measurements. We now ask whether it also has relevant operational meaning in quantum information tasks. In particular, do there exist relevant tasks where classical measurement models clearly distinguish themselves from JM? To this end, we revisit scenarios where  quantum measurements are implemented in a sequence, in such a way that the first measurement does not disturb the outcome of the second measurement. Such problems are well-studied in several variations \cite{Guhne2023}. Here, we consider scenarios, which we define as average non-disturbing, where experimenters are allowed to share classical randomness and we show that the existence of a classical measurement model provides a guarantee for such an implementation, while the same cannot be associated with JM.

\subsection{Average non-disturbance task}
We consider the problem of whether quantum measurements can be implemented in sequence without causing a disturbance on average. Specifically, consider that we are given two measurements, that we for simplicity label as $A_a=M_{a|1}$ and $B_{b}=M_{b|2}$. Our goal is to realise first $\{A_a\}_a$ and then $\{B_b\}_b$ on an arbitrary state $\rho$ in such a way that the former does not disturb the outcome statistics of the latter. This scenario is illustrated in Fig~\ref{fig_sequence}. A source of classical randomness samples $\lambda$ from some distribution $q(\lambda)$ and sends it to the two measurement devices. The two devices are then programmed to implement the POVM $\{A_{a|\lambda}\}_a$ and $\{B_{b|\lambda}\}_b$ respectively. These may be arbitrary, but their effective action must correspond to the target measurements,  namely 
\begin{equation}\label{Average_POVMs}
	\begin{aligned}
		&A_{a} = \sum_{\lambda} q(\lambda) A_{a|\lambda}, \quad \text{and} \quad B_{b} = \sum_{\lambda} q(\lambda) B_{b|\lambda}.
	\end{aligned}
\end{equation}
For a given $\lambda$, the state arriving at the second device depends not only on $\rho$ and $\{A_{a|\lambda}\}_a$ but also on the instrument that implements the latter. Denote that instrument by $\{\mathcal{I}_{a|\lambda}\}_a$. Formally, $\mathcal{I}_{a|\lambda}$ is a completely positive map, $\sum_a \mathcal{I}_{a|\lambda}$ is trace-preserving and the dual map satisfies $\mathcal{I}^\dagger_{a|\lambda}(\mathds{1})=A_{a|\lambda}$. In order to be average non-disturbing, we must have that the outcome statistics of the second measurement on the state updated by the first measurement is identical to the outcome statistics if the first measurement had not taken place. That is, $\sum_{a,\lambda}q(\lambda) \tr\left(\mathcal{I}_{a|\lambda}(\rho) B_{b|\lambda}\right) =\tr\left(\rho B_b\right)$ $\forall \rho$. This can equivalently be written as
\begin{equation}\label{NDcond}
	\sum_{a,\lambda} q(\lambda) \mathcal{I}_{a|\lambda}^\dagger(B_{b|\lambda})=B_b.
\end{equation}
If there exists instruments $\{I_{a|\lambda}\}$ and POVMs $\{B_{b|\lambda}\}$ such that the conditions \eqref{Average_POVMs} and \eqref{NDcond} are satisfied, we say that the POVMs $\{A_a\}_a$ and $\{B_b\}_b$ admit an average non-disturbing realisation.

\begin{figure}
	\includegraphics[width=0.6\columnwidth]{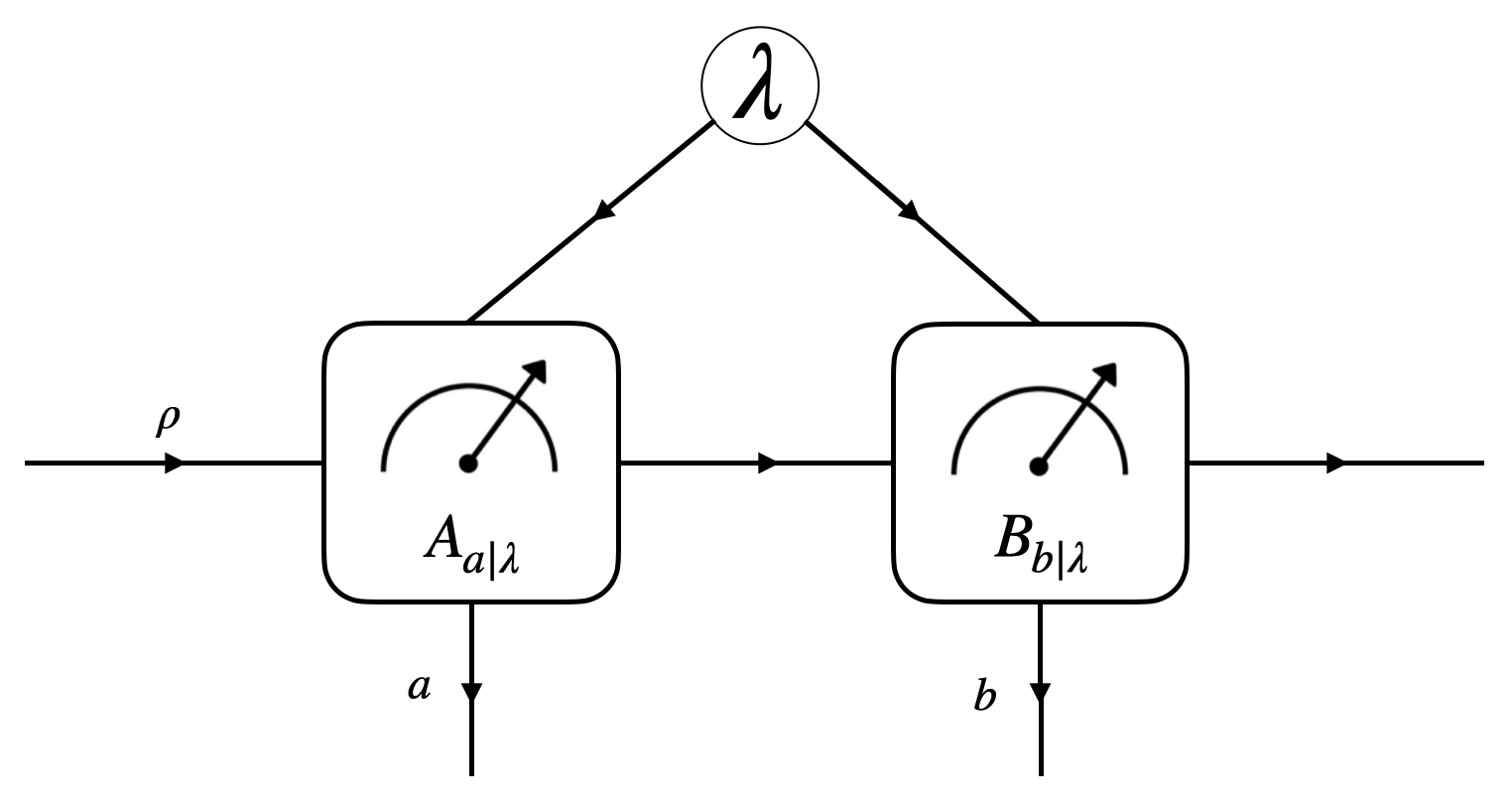}
	\caption{\textbf{Sequential measurements.} A state $\rho$ is measured twice in sequence by devices that share a classical common cause $\lambda$. They statistically implement specific POVMs $\{A_a\}$ and $\{B_b\}$. The measurements are non-disturbing if there exists a way to implement $\{A_a\}$ without altering the outcome distribution of $\{B_b\}$.} \label{fig_sequence}
\end{figure}

\subsection{Measurement classicality implies average non-disturbance}
	
Let us now consider the relationship between average non-disturbance, classical models and JM. Firstly, average non-disturbance implies JM. Indeed, consider that we employ $G_{ab} = \sum_{\lambda} q({\lambda}) \mathcal{I}_{a|\lambda}^\dagger(B_{b|\lambda})$ as the parent POVM. This is a valid POVM because the completeness of $\{G_{ab}\}_{ab}$ follows from  $\sum_b B_{b|\lambda}=\mathds{1}$ and that the dual of a CPTP map is unital. We can now check the marginals:  $\sum_b G_{ab}=A_a$ due to Eq~\eqref{Average_POVMs} and $\sum_a G_{ab}=B_b$ due to Eq~\eqref{NDcond}. 
	
Secondly, JM does not imply average non-disturbance. To show this, we use that if $\{B_b\}_b$ is extremal, only trivial decompositions of the type \eqref{Average_POVMs} are possible, i.e. $B_{b|\lambda} = B_{b}$ for each $\lambda$. Therefore, by defining $\sum_{\lambda}q(\lambda)\mathcal{I}_{a|\lambda}\coloneqq \mathcal{I}_{a}$, our problem reduces to the notion of non-disturbance studied in \cite{Heinosaari2010}, where no classical randomness is allowed. It was shown in \cite{Heinosaari2010} that, for qubits, commutation is necessary and sufficient for such non-disturbance.  For example, if  $\{A_a\}_a$ and $\{B_b\}_b$ are both the same extremal non-projective qubit measurement then they are trivially jointly measurable, but the former disturbs the latter if the measurement operators do not commute.  Thus, JM is too powerful to guarantee average non-disturbance.

Thirdly, in contrast to JM, the classicality of the measurement pair $\{A_a\}_a$ and $\{B_b\}_b$ implies non-disturbance. 
\begin{result}[Classicality implies average non-disturbance]
If  $\{A_a\}_a$ and $\{B_b\}_b$ are a pair of measurements with a classical model, then there exists an average non-disturbing sequential implementation. 
\end{result}
\begin{proof}
Following Eq~\eqref{Clmodel_def}, we  define a classical measurement model $A_a=\sum_{\lambda}q(\lambda)A_{a|\lambda}$ and $B_b=\sum_{\lambda}q(\lambda)B_{b|\lambda}$, where $[A_{a|\lambda},B_{b|\lambda}] = 0$ for each $a,b$. To prove that $\lbrace A_a, B_b \rbrace_{a,b}$ are average non-disturbing, we need to verify that there exists a collection of instruments $\{\mathcal{I}_{a|\lambda}\}_\lambda$ implementing $A_{a|\lambda}$ such that Eq~\eqref{NDcond} is satisfied. To this end, we consider the L\"uders instrument, $\mathcal{I}_{a|\lambda}(\rho) = \sqrt{A_{a|\lambda}}\rho\sqrt{A_{a|\lambda}}$ for any state $\rho$. Note that in this case, $ \mathcal{I}_{a|\lambda}^{\dagger}(M) = \sqrt{A_{a|\lambda}}M \sqrt{A_{a|\lambda}}$ for any measurement $M$. Therefore, 
\begin{equation}
	\sum_{a,\lambda} q(\lambda) I_{a|\lambda}^\dagger(B_{b|\lambda})  =\sum_{\lambda} q(\lambda) \, \sum_{a} \sqrt{A_{a|\lambda}} B_{b|\lambda} \sqrt{A_{a|\lambda}} = \sum_{\lambda} q(\lambda) \, \sum_{a} A_{a|\lambda} B_{b|\lambda} = B_{b},
\end{equation}
where we have used that, since $\lbrace A_{a|\lambda} \rbrace_{a,\lambda}$ are POVMs, then $[\sqrt{A_{a|\lambda}},B_{b|\lambda}] = 0$ and $\sum_a A_{a|\lambda} = \openone$.
\end{proof}
In view of this result, we conclude  the following chain of implications to hold:
\begin{equation}
	\text{Classical model} \quad \Rightarrow \quad\text{Average non-disturbance} \quad \Rightarrow \quad\text{JM}.
\end{equation}
For unbiased dichotomic qubit measurements, this hierarchy collapses and all three concepts become equivalent. In general, however, the implications are strict. The fact that non-disturbance does not imply classicality is expected, because the former is not a symmetric concept. In other words, swapping the measurement order could change whether $\{A_a\}_a$ and $\{B_b\}_b$ admit a non-disturbing implementation.  In Supplementary Material we explicitly show that non-disturbance does not imply classicality for POVMs in dimension $d \geq 5$.

\section{Discussion}\label{Discussion}
Quantum measurement devices routinely rely on controlling superposition properties.  Here, we have proposed conceptually classical models for simulating quantum measurements, centered on the core idea that no device employed in the simulation model can have any superposition features. We determined the exact noise and loss thresholds at which all projective $d$-dimensional measurements admit a classical model. Interestingly, both these thresholds coincide with those at which the set of all von Neumann measurements is jointly measurable \cite{Uola2014, Sekatski2024}. This means that even though joint measurability is a strictly larger set, the two notions become equivalent in the limit in which all projective measurements allowed in the theory are considered. It remains an interesting open problem to identify this threshold when also including non-projective measurements. Since classical models correspond to JM with parent POVMs which are projectively simulable, lower bounds on this threshold can be found in \cite{Barrett2002, Almeida2007}, although they are not expected to be optimal. In addition, we have seen that classical measurements are operationally distinct from joint measurability, as the former was sufficient for succeeding with sequential non-disturbance tasks, whereas the latter was only necessary. A deeper understanding of the relation between these two concepts may be unveiled by studying whether the gap between them can grow, perhaps in an unbounded way, by considering non-disturbance tasks with more than two sequential measurements.

As we have previously mentioned, the present work can be seen as a continuation of the research line started with the introduction of classical models for simulating quantum states in Ref~\cite{Cobucci2025}. We now discuss the relationship between the two problems; see Supplementary Material for extended details. Firstly, as already pointed out, the case of  states can be seen as a special case of measurements, but importantly not vice versa. Secondly, classical measurement models in Eq~\eqref{Clmodel_predef} include a non-trivial post-processing which classical state models do not. Although Definition \ref{Def1} shows that the post-processing can be eliminated, it comes at the price of allowing arbitrary ranks in the measurement devices. This has no counterpart for classical states. Thirdly, while the argument leading up to the part of Result~\ref{res2} concerning isotropic noise admits a clear connection to the analysis of classical state models for all pure states in quantum theory (see Ref~\cite{Cobucci2025}), no such correspondence appears relevant for the part of Result~\ref{res2} that concerns particle losses. This is because measurements represent a more general situation, where also additional forms of decoherence become relevant. Lastly, for classical state models Ref~\cite{Cobucci2025} reports connections both to quantum steering and to criteria for determining that states cannot give rise to quantum correlations in prepare-and-measure scenarios. While these connections are non-trivial for states, their counterparts for measurements are much less insightful: classical measurement models imply JM, and JM is known to imply both the absence of steering and of quantum correlations in prepare-and-measure scenarios \cite{Frenkel_2015}. Instead, we have seen that more interesting operational interpretations for classical measurement models emerge in very different quantum information settings, namely  non-disturbance tasks. Our work leaves the natural open problem of developing classical models for the most general quantum operations, namely quantum channels.

\begin{acknowledgements}
	We thank Paul Skrzypczyk for useful discussions and Martin J. Renner for useful comments. This work is supported by the Knut and Alice Wallenberg Foundation through the Wallenberg Center for Quantum Technology (WACQT), the Wallenberg Initiative on Networks and Quantum Information (WINQ), the Swedish Research Council under Contract No.~2023-03498 and 2024-05341, the Crafoord Foundation, the Krapperup Foundation and the Swedish Foundation for Strategic Research. G.C. acknowledges Nordita visiting PhD fellowship.		
\end{acknowledgements}
	
	\bibliography{references_clMeasurement}

@article{Guhne2023,
	title = {Colloquium: Incompatible measurements in quantum information science},
	author = {G\"uhne, Otfried and Haapasalo, Erkka and Kraft, Tristan and Pellonp\"a\"a, Juha-Pekka and Uola, Roope},
	journal = {Rev. Mod. Phys.},
	volume = {95},
	issue = {1},
	pages = {011003},
	numpages = {25},
	year = {2023},
	month = {Feb},
	publisher = {American Physical Society},
	doi = {10.1103/RevModPhys.95.011003},
	url = {https://link.aps.org/doi/10.1103/RevModPhys.95.011003}
}

@article{Wolf2009,
	title = {Measurements Incompatible in Quantum Theory Cannot Be Measured Jointly in Any Other No-Signaling Theory},
	author = {Wolf, Michael M. and Perez-Garcia, David and Fernandez, Carlos},
	journal = {Phys. Rev. Lett.},
	volume = {103},
	issue = {23},
	pages = {230402},
	numpages = {4},
	year = {2009},
	month = {Dec},
	publisher = {American Physical Society},
	doi = {10.1103/PhysRevLett.103.230402},
	url = {https://link.aps.org/doi/10.1103/PhysRevLett.103.230402}
}

@article{Quintino2014,
	title = {Joint Measurability, Einstein-Podolsky-Rosen Steering, and Bell Nonlocality},
	author = {Quintino, Marco T\'ulio and V\'ertesi, Tam\'as and Brunner, Nicolas},
	journal = {Phys. Rev. Lett.},
	volume = {113},
	issue = {16},
	pages = {160402},
	numpages = {5},
	year = {2014},
	month = {Oct},
	publisher = {American Physical Society},
	doi = {10.1103/PhysRevLett.113.160402},
	url = {https://link.aps.org/doi/10.1103/PhysRevLett.113.160402}
}

@article{Oszmaniec2017,
	title = {Simulating Positive-Operator-Valued Measures with Projective Measurements},
	author = {Oszmaniec, Micha\l{} and Guerini, Leonardo and Wittek, Peter and Ac\'{\i}n, Antonio},
	journal = {Phys. Rev. Lett.},
	volume = {119},
	issue = {19},
	pages = {190501},
	numpages = {6},
	year = {2017},
	month = {Nov},
	publisher = {American Physical Society},
	doi = {10.1103/PhysRevLett.119.190501},
	url = {https://link.aps.org/doi/10.1103/PhysRevLett.119.190501}
}

@misc{Brask2026,
	title={Quantum correlations in prepare-and-measure scenarios and their semi-device-independent applications}, 
	author={Jonatan Bohr Brask and Nicolas Brunner and Jef Pauwels and Davide Rusca and Armin Tavakoli},
	year={2026},
	eprint={2603.23604},
	archivePrefix={arXiv},
	primaryClass={quant-ph},
	url={https://arxiv.org/abs/2603.23604}, 
}

@article{Uola2015,
	title = {One-to-One Mapping between Steering and Joint Measurability Problems},
	author = {Uola, Roope and Budroni, Costantino and G\"uhne, Otfried and Pellonp\"a\"a, Juha-Pekka},
	journal = {Phys. Rev. Lett.},
	volume = {115},
	issue = {23},
	pages = {230402},
	numpages = {5},
	year = {2015},
	month = {Dec},
	publisher = {American Physical Society},
	doi = {10.1103/PhysRevLett.115.230402},
	url = {https://link.aps.org/doi/10.1103/PhysRevLett.115.230402}
}

@misc{Brinster2025,
	title={Robust certification of non-projective measurements: theory and experiment}, 
	author={Raphael Brinster and Peter Tirler and Shishir Khandelwal and Michael Meth and Hermann Kampermann and Dagmar Bruß and Rainer Blatt and Martin Ringbauer and Armin Tavakoli and Nikolai Wyderka},
	year={2025},
	eprint={2511.04446},
	archivePrefix={arXiv},
	primaryClass={quant-ph},
	url={https://arxiv.org/abs/2511.04446}, 
}

@article{Frenkel_2015,
	title={Classical Information Storage in an n-Level Quantum System},
	volume={340},
	ISSN={1432-0916},
	url={http://dx.doi.org/10.1007/s00220-015-2463-0},
	DOI={10.1007/s00220-015-2463-0},
	number={2},
	journal={Communications in Mathematical Physics},
	publisher={Springer Science and Business Media LLC},
	author={Frenkel, Péter E. and Weiner, Mihály},
	year={2015},
	month=Sept, pages={563–574} }

@article{Busch1986,
	title = {Unsharp reality and joint measurements for spin observables},
	author = {Busch, Paul},
	journal = {Phys. Rev. D},
	volume = {33},
	issue = {8},
	pages = {2253--2261},
	numpages = {0},
	year = {1986},
	month = {Apr},
	publisher = {American Physical Society},
	doi = {10.1103/PhysRevD.33.2253},
	url = {https://link.aps.org/doi/10.1103/PhysRevD.33.2253}
}

@article{Brougham2007,
	title = {Estimating the expectation values of spin-1/2 observables with finite resources},
	author = {Brougham, Thomas and Andersson, Erika},
	journal = {Phys. Rev. A},
	volume = {76},
	issue = {5},
	pages = {052313},
	numpages = {10},
	year = {2007},
	month = {Nov},
	publisher = {American Physical Society},
	doi = {10.1103/PhysRevA.76.052313},
	url = {https://link.aps.org/doi/10.1103/PhysRevA.76.052313}
}

@incollection{Bell1977,
	author    = {Bell, John Stewart},
	title     = {Free Variables and Local Causality},
	booktitle = {Speakable and Unspeakable in Quantum Mechanics},
	publisher = {Cambridge University Press},
	year      = {2004},
	edition   = {2nd},
	pages     = {100--104},
	chapter   = {12},
	note      = {Originally published in \textit{Epistemological Letters}, vol. 15, pp. 79--84 (1977)}
}

@article{DAriano2005,
	author  = {D'Ariano, Giacomo Mauro and Lo Presti, Paoloplacido and Perinotti, Paolo},
	title   = {Classical randomness in quantum measurements},
	journal = {Journal of Physics A: Mathematical and General},
	volume  = {38},
	number  = {26},
	pages   = {5979--5991},
	year    = {2005},
	doi     = {10.1088/0305-4470/38/26/010},
	eprint  = {quant-ph/0408115},
	archivePrefix = {arXiv}
}

@Article{Cobucci2025,
	author={Cobucci, Gabriele
	and Bernal, Alexander
	and Renner, Martin J.
	and Tavakoli, Armin},
	title={Operationally classical simulation of quantum states},
	journal={Nature Communications},
	year={2026},
	month={Jan},
	day={27},
	volume={17},
	number={1},
	pages={1104},
	issn={2041-1723},
	doi={10.1038/s41467-026-68581-3},
	url={https://doi.org/10.1038/s41467-026-68581-3}
}

@article{Uola_2020,
   title={Quantum steering},
   volume={92},
   ISSN={1539-0756},
   url={http://dx.doi.org/10.1103/RevModPhys.92.015001},
   DOI={10.1103/revmodphys.92.015001},
   number={1},
   journal={Reviews of Modern Physics},
   publisher={American Physical Society (APS)},
   author={Uola, Roope and Costa, Ana C. S. and Nguyen, H. Chau and Gühne, Otfried},
   year={2020},
   month=mar }

@article{Durt2010,
author = {Durt, Thomas and Englert, Berthold-Georg and Bengtsson, Ingemar and \.{Z}yczkowski, Karol},
title = {On mutually unbiased bases},
journal = {International Journal of Quantum Information},
volume = {08},
number = {04},
pages = {535-640},
year = {2010},
doi = {10.1142/S0219749910006502},

URL = { 
    
        https://doi.org/10.1142/S0219749910006502
    
    

},
eprint = { 
    
        https://doi.org/10.1142/S0219749910006502
    
    

}
}

@article{Renes2004,
   title={Symmetric informationally complete quantum measurements},
   volume={45},
   ISSN={1089-7658},
   url={http://dx.doi.org/10.1063/1.1737053},
   DOI={10.1063/1.1737053},
   number={6},
   journal={Journal of Mathematical Physics},
   publisher={AIP Publishing},
   author={Renes, Joseph M. and Blume-Kohout, Robin and Scott, A. J. and Caves, Carlton M.},
   year={2004},
   month=jun, pages={2171–2180} }

@article{Tavakoli2024,
   title={Semidefinite programming relaxations for quantum correlations},
   volume={96},
   ISSN={1539-0756},
   url={http://dx.doi.org/10.1103/RevModPhys.96.045006},
   DOI={10.1103/revmodphys.96.045006},
   number={4},
   journal={Reviews of Modern Physics},
   publisher={American Physical Society (APS)},
   author={Tavakoli, Armin and Pozas-Kerstjens, Alejandro and Brown, Peter and Araújo, Mateus},
   year={2024},
   month=dec }

@article{Heinosaari2010,
   title={Nondisturbing quantum measurements},
   volume={51},
   ISSN={1089-7658},
   url={http://dx.doi.org/10.1063/1.3480658},
   DOI={10.1063/1.3480658},
   number={9},
   journal={Journal of Mathematical Physics},
   publisher={AIP Publishing},
   author={Heinosaari, Teiko and Wolf, Michael M.},
   year={2010},
   month=sep }

@article{Sekatski2024,
   title={Compatibility of projective measurements subject to white noise and loss},
   volume={109},
   ISSN={2469-9934},
   url={http://dx.doi.org/10.1103/PhysRevA.109.022215},
   DOI={10.1103/physreva.109.022215},
   number={2},
   journal={Physical Review A},
   publisher={American Physical Society (APS)},
   author={Sekatski, Pavel},
   year={2024},
   month=feb }

@article{Peres1990,
	author = {Asher Peres},
	doi = {10.1007/bf01883517},
	journal = {Foundations of Physics},
	number = {12},
	pages = {1441--1453},
	title = {Neumark's Theorem and Quantum Inseparability},
	volume = {20},
	year = {1990}
}

@article{Streltsov2017,
 	title = {Colloquium: Quantum coherence as a resource},
 	author = {Streltsov, Alexander and Adesso, Gerardo and Plenio, Martin B.},
 	journal = {Rev. Mod. Phys.},
 	volume = {89},
 	issue = {4},
 	pages = {041003},
 	numpages = {34},
 	year = {2017},
 	month = {Oct},
 	publisher = {American Physical Society},
 	doi = {10.1103/RevModPhys.89.041003},
 	url = {https://link.aps.org/doi/10.1103/RevModPhys.89.041003}
 }

@misc{cobucci2025MaxNP,
      title={Maximally non-projective measurements are not always symmetric informationally complete}, 
      author={Gabriele Cobucci and Raphael Brinster and Shishir Khandelwal and Hermann Kampermann and Dagmar Bruß and Nikolai Wyderka and Armin Tavakoli},
      year={2025},
      eprint={2508.03652},
      archivePrefix={arXiv},
      primaryClass={quant-ph},
      url={https://arxiv.org/abs/2508.03652}, 
}

@article{Khandelwal2025,
   title={Simulating Quantum Instruments with Projective Measurements and Quantum Postprocessing},
   volume={135},
   ISSN={1079-7114},
   url={http://dx.doi.org/10.1103/bhr5-g71p},
   DOI={10.1103/bhr5-g71p},
   number={4},
   journal={Physical Review Letters},
   publisher={American Physical Society (APS)},
   author={Khandelwal, Shishir and Tavakoli, Armin},
   year={2025},
   month=jul }

@misc{github-code_simu-witness,
	author = {},
 	title = {Code for classical simulation of sets of quantum measurements and witness},
	note = {\url{https://github.com/GabrieleCobucci/classical_measurement_sets}},
	year = 2026,
}

@misc{proof_MUB,
	note = {Let $A$ and $B$ be two unbiased projectors. Since $A+B$ has
	rank two, denote its non-zero eigenvalues by $e_1$ and $e_2$,
	with $e_1 \geq e_2$. They satisfy
	$e_1+e_2=\operatorname{tr}(A+B)=2$. Moreover,
	$e_1^2+e_2^2=\operatorname{tr}((A+B)^2)
	=2+2\operatorname{tr}(AB)=2+\frac{2}{d}$,
	where the last equality follows from $A$ and $B$ being unbiased.
	Solving these equations gives
	$e_1=1+\frac{1}{\sqrt{d}}$ and
	$e_2=1-\frac{1}{\sqrt{d}}$.}
}

@article{Mandayam2010,
	title={A transform of complementary aspects with applications to entropic uncertainty relations},
	volume={51},
	ISSN={1089-7658},
	url={http://dx.doi.org/10.1063/1.3477319},
	DOI={10.1063/1.3477319},
	number={8},
	journal={Journal of Mathematical Physics},
	publisher={AIP Publishing},
	author={Mandayam, Prabha and Wehner, Stephanie and Balachandran, Niranjan},
	year={2010},
	month=Aug }

@article{Rastegin2013,
	title={Uncertainty relations for MUBs and SIC-POVMs in terms of generalized entropies},
	volume={67},
	ISSN={1434-6079},
	url={http://dx.doi.org/10.1140/epjd/e2013-40453-2},
	DOI={10.1140/epjd/e2013-40453-2},
	number={12},
	journal={The European Physical Journal D},
	publisher={Springer Science and Business Media LLC},
	author={Rastegin, Alexey E.},
	year={2013},
	month=Dec }

@misc{five_tetrahedra,
 	title = {Compound of five tetrahedra},
	note = {\url{https://polytope.miraheze.org/wiki/Compound_of_five_tetrahedra}},
	year = 2025,
}

@article{Uola2014,
   title={Joint Measurability of Generalized Measurements Implies Classicality},
   volume={113},
   ISSN={1079-7114},
   url={http://dx.doi.org/10.1103/PhysRevLett.113.160403},
   DOI={10.1103/physrevlett.113.160403},
   number={16},
   journal={Physical Review Letters},
   publisher={American Physical Society (APS)},
   author={Uola, Roope and Moroder, Tobias and Gühne, Otfried},
   year={2014},
   month=oct }

@article{Guerini2017,
   title={Operational framework for quantum measurement simulability},
   volume={58},
   ISSN={1089-7658},
   url={http://dx.doi.org/10.1063/1.4994303},
   DOI={10.1063/1.4994303},
   number={9},
   journal={Journal of Mathematical Physics},
   publisher={AIP Publishing},
   author={Guerini, Leonardo and Bavaresco, Jessica and Terra Cunha, Marcelo and Acín, Antonio},
   year={2017},
   month=sep }

@article{Jones07,
  title = {Entanglement, Einstein-Podolsky-Rosen correlations, Bell nonlocality, and steering},
  author = {Jones, S. J. and Wiseman, H. M. and Doherty, A. C.},
  journal = {Phys. Rev. A},
  volume = {76},
  issue = {5},
  pages = {052116},
  numpages = {18},
  year = {2007},
  month = {Nov},
  publisher = {American Physical Society},
  doi = {10.1103/PhysRevA.76.052116},
  url = {https://link.aps.org/doi/10.1103/PhysRevA.76.052116}
}

@article{Werner89,
  title = {Quantum states with Einstein-Podolsky-Rosen correlations admitting a hidden-variable model},
  author = {Werner, Reinhard F.},
  journal = {Phys. Rev. A},
  volume = {40},
  issue = {8},
  pages = {4277--4281},
  numpages = {0},
  year = {1989},
  month = {Oct},
  publisher = {American Physical Society},
  doi = {10.1103/PhysRevA.40.4277},
  url = {https://link.aps.org/doi/10.1103/PhysRevA.40.4277}
}

@article{TaylorMarkov,
 ISSN = {00029890, 19300972},
 URL = {http://www.jstor.org/stable/2311462},
 author = {A. S. Davis},
 journal = {The American Mathematical Monthly},
 number = {3},
 pages = {264--267},
 publisher = {[Taylor & Francis, Ltd., Mathematical Association of America]},
 title = {Markov Chains as Random Input Automata},
 urldate = {2026-02-19},
 volume = {68},
 year = {1961}
}

@article{Renner2025,
  title = {Compatibility of Generalized Noisy Qubit Measurements},
  author = {Renner, Martin J.},
  journal = {Phys. Rev. Lett.},
  volume = {132},
  issue = {25},
  pages = {250202},
  numpages = {8},
  year = {2024},
  month = {Jun},
  publisher = {American Physical Society},
  doi = {10.1103/PhysRevLett.132.250202},
  url = {https://link.aps.org/doi/10.1103/PhysRevLett.132.250202}
}

@article{Yujie2024,
  title = {Exact Steering Bound for Two-Qubit Werner States},
  author = {Zhang, Yujie and Chitambar, Eric},
  journal = {Phys. Rev. Lett.},
  volume = {132},
  issue = {25},
  pages = {250201},
  numpages = {7},
  year = {2024},
  month = {Jun},
  publisher = {American Physical Society},
  doi = {10.1103/PhysRevLett.132.250201},
  url = {https://link.aps.org/doi/10.1103/PhysRevLett.132.250201}
}

@article{Barrett2002,
   title={Nonsequential positive-operator-valued measurements on entangled mixed states do not always violate a Bell inequality},
   volume={65},
   ISSN={1094-1622},
   url={http://dx.doi.org/10.1103/PhysRevA.65.042302},
   DOI={10.1103/physreva.65.042302},
   number={4},
   journal={Physical Review A},
   publisher={American Physical Society (APS)},
   author={Barrett, Jonathan},
   year={2002},
   month=mar }

@article{Almeida2007,
   title={Noise Robustness of the Nonlocality of Entangled Quantum States},
   volume={99},
   ISSN={1079-7114},
   url={http://dx.doi.org/10.1103/PhysRevLett.99.040403},
   DOI={10.1103/physrevlett.99.040403},
   number={4},
   journal={Physical Review Letters},
   publisher={American Physical Society (APS)},
   author={Almeida, Mafalda L. and Pironio, Stefano and Barrett, Jonathan and Tóth, Géza and Acín, Antonio},
   year={2007},
   month=jul }
	
	\appendix
		\section{Exact noise and loss rates for simulating all projective measurements }\label{Exact_noise_loss_proj}
	In this appendix, we provide a classical model for simulating all projective measurements.  Let $\mathbf{M}_{v}$ be the set of all rank-1 measurements of a $d$-dimensional Hilbert space subject to the white noise $v\in [0,1]$. We can parametrize it considering the $d$-dimensional computational basis $\lbrace \ket{a} \rbrace_{a=1}^{d}$ and all possible unitaries $U \in \text{SU}(d)$,
	\begin{equation}
		\mathbf{M}_{v} = \left\lbrace N_{a|U} = v\, U\ketbra{a}{a}U^{\dagger} + \frac{1-v}{d}\openone \right\rbrace_{U\in \text{SU}(d)}.
	\end{equation}
	Let us first consider no losses. We provide a classical simulation of $\mathbf{M}_{v}$ up to $v_{\text{sim}}=\frac{H_d-1}{d-1}$.  Consider the following classical model
	\begin{equation}\label{App_eq:model_M}
		\int d\mu_{\text{Haar}}(V) \sum_{k=1}^d p^\star (a|U, k, V)\, V\ketbra{k} V^\dagger,
	\end{equation}
	in which we have replaced the pre-processing distribution over $\lbrace q(\lambda) \rbrace_{\lambda}$ in \eqref{App_Measurement_simulability} with the integral over the Haar measure of all possible unitaries $V$. Furthermore, we defined a specific post-processing rule $p^\star(a|U, k, V)$ as
	\begin{equation}\label{probability_simulation_specific}
		p^\star(a|U, k, V)=\left\{
		\begin{aligned}
			1&\quad \text{if }\ \abs{\bra{k}V^\dagger U \ket{a}}^2>\abs{\bra{k}V^\dagger U \ket{a'}}^2\quad \forall a' \neq a,\\
			0&\quad \text{otherwise.}
		\end{aligned}\right.
	\end{equation}
	Let $W$ be any unitary matrix diagonal in the $\{U\ket{a}\}_{a}$ basis, i.e. verifying $W (U\ket{a})= e^{i \theta_a}(U\ket{a})$.  Then, the matrix $T$ resulting from the integration \eqref{App_eq:model_M} is invariant under $W$ . Indeed, 
	\begin{equation}
		W T W^\dagger=\int d\mu_{\text{Haar}}(V) \sum_{k=1}^d p^\star (a|U, k, V)\, W V\ketbra{k}{k} V^\dagger W^\dagger=\int d\mu_{\text{Haar}}(\tilde V) \sum_{k=1}^d p^\star (a|U, k, W^\dagger \tilde V)\,  \tilde V\ketbra{k}{k} {\tilde V}^\dagger
	\end{equation}
	where we have defined $\tilde V= W V$ and used the invariance of the Haar measure $d\mu_{\text{Haar}}(V)=d\mu_{\text{Haar}}(W^\dagger \tilde V)=d\mu_{\text{Haar}}(\tilde V) $.  Notice that all absolute values $ \{|\bracket{k}{(W^\dagger \tilde V)^{\dagger}U}{a'}|^2\}_{a'}$ remain invariant since $|\bracket{k}{(W^\dagger \tilde V)^{\dagger}U}{a'}|^2 = |\bracket{k}{{\tilde V}^{\dagger} W U}{a'}|^2 = |e^{i \theta_{a'}} \bracket{k}{{\tilde V}^{\dagger}U}{a'}|^2= | \bracket{k}{{\tilde V}^{\dagger}U}{a'}|^2$. Hence,  from the definition \eqref{probability_simulation_specific} of the post-processing rule, we have $p^{\star}(a|U,k,W^\dagger \tilde V)=p^{\star}(a|U,k,{\tilde V})$ and  
	\begin{equation}
		W T W^\dagger=\int d\mu_{\text{Haar}}(\tilde V) \sum_{k=1}^d p^\star (a|U, k, \tilde V)\,  \tilde V\ketbra{k}{k} {\tilde V}^\dagger=T.
	\end{equation}
	Since $T$ commutes with any diagonal unitary matrix in the $\{U\ket{a}\}_{a}$ basis, it is direct that $T$ itself is diagonal in that same basis.  Now, consider the set of permutations $\{P_\pi^{(a)}\}_{\pi}$ of the elements in $\{U\ket{a'}\}_{a'}$ leaving invariant the $U\ket{a}$ state.  The integral $T$ is also invariant under these kinds of transformations. To see this, we follow a similar reasoning as in the previous derivation and notice that if $\abs{\bra{k}V^\dagger U \ket{a}}^2 >\abs{\bra{k}V^\dagger U \ket{a'}}^2\quad \forall a' \neq a$, then $\abs*{\bra{k}V^\dagger P_\pi^{(a)} U \ket{a}}^2=\abs*{\bra{k}V^\dagger U \ket{a}}^2 >\abs*{\bra{k}V^\dagger P_\pi^{(a)} U \ket{a'}}^2\quad \forall a'\neq a$, since the permutation $P_\pi^{(a)}$ only reshuffles the $\{U\ket{a'}\}_{a'\neq a}$ set.
	
	In consequence, $T$ is diagonal in the $\{U\ket{a}\}_{a}$ basis and the entries associated with any $a'\neq a$ have to be all equal among them. Using that $\tr(T)=1$, we conclude that $T$ must yield a measurement of the form
	\begin{equation}
		\int d\mu_{\text{Haar}}(V) \sum_{k=1}^d p^\star (a|U, k, V)\, V\ketbra{k} V^\dagger=v\, U\ketbra*{a}U^\dagger+\frac{1-v}{d} \openone,
	\end{equation}
	for some visibility $v \in [0,1]$. In order to compute this visibility, we bracket both sides by $\bra{a}U^\dagger \cdot U \ket{a}$:
	\begin{equation}
		\frac{(d-1)v+1}{d}=\sum_{k=1}^d \int d\mu_{\text{Haar}}(V) p^\star (a|U, k, V)\, \abs{\bra{k}V^\dagger U \ket{a}}^2.
	\end{equation}
	Using again the invariance of the Haar measure, we can replace $(U,V)$ with a single unitary dependence $V$. Therefore,
	\begin{equation}
		\frac{(d-1)v+1}{d}=\sum_{k=1}^d \int d\mu_{\text{Haar}}(V) p^\star (a|k,V)\, \abs{\mel{k}{V}{a}}^2.
	\end{equation}
	For each value of $k$, the integral to compute is of the form:
	\begin{equation}
		\int d\mu_{\text{Haar}}(V) p^\star (a|k,V)\, \abs{\mel{k}{V}{a}}^2=\int_a d\mu_{\text{Haar}}(V)\, \abs{\mel{k}{V}{a}}^2,
	\end{equation}
	where the subscript $a$ indicates that the integration is only over unitaries $V$ such that $\abs{\mel{k}{V}{a}}^2$ is greater than $\abs{\mel{k}{V}{a'}}^2$ for any other $a'\neq a$. This kind of integral has already been studied in \cite{Werner89,Jones07,Cobucci2025} and its value, independently on $k$,  reads
	\begin{equation}
		\int_a d\mu_{\text{Haar}}(V)\, \abs{\mel{k}{V}{a}}^2=\frac{1}{d}\frac{H_d}{d}.
	\end{equation}
	Thus, 
	\begin{equation}
		\sum_{k=1}^d \int d\mu_{\text{Haar}}(V) p^\star (a|k,V)\, \abs{\mel{k}{V}{a}}^2=\sum_{k=1}^d \frac{1}{d}\frac{H_d}{d}=\frac{H_d}{d}\quad \implies \quad v_{\text{sim}}=\frac{H_d-1}{d-1}.
	\end{equation}
	Therefore,  for any $v\leq v_{\text{sim}}$ the set $\mathbf{M}_v$ is classically simulable.

	It is possible to modify the previous simulation to also account for a loss rate $\eta$. Notice that we parametrize any measurement $\{M_a\}_a$ experiencing a non-zero loss via the loss channel $L_\eta(M_a)$. Here, the loss channel implements the original measurement with a probability $\eta$, while assigning the outcome $M_{\varnothing}=(1-\eta)\openone$ otherwise:
	\begin{equation}\label{measurements_loss}
		M_{a|x}^{(v,\eta)}=\left\{
		\begin{aligned}
			\eta M_a|x&\quad \text{if }\ a=1,\dots,d\\
			(1-\eta)\openone&\quad a=\varnothing.
		\end{aligned}\right.
	\end{equation}
	The corresponding response function $p^\star(a|U, k, V)$ to use in the simulation now reads:
	\begin{equation}\label{probability_simulation_specific_loss}
		p^\star(a|U, k, V)=\left\{
		\begin{aligned}
			1&\quad \text{if }\ a=a^\star(U,k,V)\\
			0&\quad \text{otherwise,}
		\end{aligned}\right.
	\end{equation}
	with 
	\begin{equation}\label{response_function}
		a^\star(U,k,V)=\left\{
		\begin{aligned}
			&\text{argmax} |\mel{k}{V^\dagger U}{a}|^2\quad& \text{if }\ \max_{a}|\mel{k}{V^\dagger U}{a}|^2\geq t\\
			&\varnothing\quad& \text{otherwise.}
		\end{aligned}\right.
	\end{equation}
	This coincides with the response function used in \cite{Sekatski2024} and hence leads to the same values of both noise and loss rates \cite{Sekatski2024}:
	\begin{equation}
		(v(t),\eta(t))=\left(t-\frac{1}{d-1}\left(\frac{S_1(t)}{1-S_0(t)}+1\right),1-S_0(t)\right),
	\end{equation}
	with
	\begin{equation}
		S_n(t)=\sum_{m=n}^{\min(\floor{\frac{1}{t}},d)}\frac{(-1)^m }{m^n}{d \choose m}(1-tm)^{d-1}.
	\end{equation}
	When there is no loss, $S_0(t)=0$. This corresponds to the case in which $t\leq 1/d$ or equivalently to $\floor{1/t}\geq d$:
	\begin{equation}
		S_0(t)=\sum_{m=0}^{d}(-1)^m {d \choose m}(1-tm)^{d-1}=0.
	\end{equation}
	The expression for $S_1(t)$ in that case reads
	\begin{equation}
		S_1(t)=\sum_{m=1}^{d}\frac{(-1)^m }{m} {d \choose m}(1-tm)^{d-1}=(d-1)t-H_d,
	\end{equation}
	which gives the previously displayed critical visibility 
	\begin{equation}
		v=\frac{H_d-1}{d-1}.
	\end{equation}
	Another interesting case comes when $t>1/2$, since then only the $m=0,1$ terms contribute:
	\begin{equation}
		S_0(t)=1-d(1-t)^{(d-1)},\quad S_1(t)=-d(1-t)^{(d-1)}.
	\end{equation}
	Therefore,
	\begin{equation}
		(v(t),\eta(t))=\left(t,d(1-t)^{(d-1)}\right).
	\end{equation}

	\section{Relation between classical simulation of measurements and classical simulation of states}\label{Relation_classical_states}
	
	In this appendix, we prove different relations between the classical simulation model for measurements introduced in the main text and the classical simulation model for states in \cite{Cobucci2025}.
	
	For clarity, we summarize the results as follows:
	\begin{enumerate}
		\item The measurements $\mathbf{M}=\{M_{a|x}: \tr(M_{a|x}) = 1\}_{a,x}$ are classically simulable $ \quad \centernot \Longrightarrow \quad$ The states $\mathcal{E}=\{\rho_{(a, x)}\equiv M_{a|x}\}_{a,x}$ are classically simulable.
		
		\item The measurements $\mathbf{M}=\{M_{a|x}: \tr(M_{a|x}) \neq 1\}_{a,x}$ are classically simulable $ \quad \centernot \Longrightarrow \quad$ The states $\mathcal{E}=\{\rho_{(a, x)}\equiv M_{a|x}/\tr(M_{a|x})\}_{a,x}$ are classically simulable.
		
		\item The states $\mathcal{E} = \lbrace \rho_{x} \rbrace_{x=1}^{m}$ are classically simulable $\quad \implies \quad$ The measurements $\mathbf{M} = \lbrace \rho_x, \openone - \rho_x \rbrace_{x=1}^{m}$ are classically simulable.  On the other hand, $\mathbf{M} = \lbrace \rho_x, \openone - \rho_x \rbrace_{x=1}^{m}$ is classically simulable $ \quad \centernot \Longrightarrow \quad$ $\mathcal{E} = \lbrace \rho_{x} \rbrace_{x=1}^{m}$ is classically simulable.
		
		\item If $d=2$, then $\mathbf{M}=\{M_{a|x}: \tr(M_{a|x})=1\}_{a,x}$ is classically simulable $ \quad \iff \quad$  $\mathcal{E}=\{\rho_{(a, x)}\equiv M_{a|x}\}_{a,x}$ is classically simulable.
	\end{enumerate}

	\subsection{Proof of implication 1}
	For the case in which $\tr(M_{a|x}) = 1$, consider the computational and Fourier bases $C_a\equiv\ketbra{a}$, $F_a\equiv\ketbra{f_a}$ in dimension $d=4$. Let us define two four-outcome POVMs $\{M_{a|x}\}_{a=1}^4$ for $x=1,2$ by
	\begin{equation}
		\begin{aligned}
			M_{1|1}&=\frac{1}{2}(C_1+C_2),&
			M_{2|1}&=\frac{1}{2}(C_3+C_4),\\
			M_{3|1}&=\frac{1}{2}(F_1+F_2),&
			M_{4|1}&=\frac{1}{2}(F_3+F_4),
		\end{aligned}
	\end{equation}
	and
	\begin{equation}
		\begin{aligned}
			M_{1|2}&=\frac{1}{2}(C_1+C_3),&
			M_{2|2}&=\frac{1}{2}(C_2+C_4),\\
			M_{3|2}&=\frac{1}{2}(F_1+F_3),&
			M_{4|2}&=\frac{1}{2}(F_2+F_4).
		\end{aligned}
	\end{equation}
	Clearly, $\tr(M_{a|x})=1$ for all $a,x$. A classical model for measurements can be built by considering two measurement devices $\lambda\in\{C, F\}$ which operate in the computational and the Fourier basis, respectively. We call the devices with equal probability $q(\lambda)=1/2$ and fix the following deterministic strategies for the post-processing rule $p(a|x, k, \lambda)$:
	
	
	
	\begin{equation}
		\begin{aligned}
			p(1|1, k, C)&=1,\quad \text{for} \quad k=1,2,  \\
			p(2|1, k, C)&=1,\quad \text{for} \quad k=3,4,  \\
			p(a|1, k, C)&=0,\quad \text{for} \quad a=3,4,  \quad \forall k,\\
			p(3|1, k, F)&=1,\quad \text{for} \quad k=1,2,  \\
			p(4|1, k, F)&=1,\quad \text{for} \quad k=3,4,  \\
			p(a|1, k, F)&=0,\quad \text{for} \quad a=1,2,  \quad \forall k,
		\end{aligned}
	\end{equation}
	and similarly for $x=2$.  Let us consider $\rho_{(a,x)}=M_{a|x}$ and let us see that no classical simulation model for states is possible. Indeed, if  
	a simulation exists:
	\begin{equation}
		\rho_{(a,x)}= \sum_{\lambda} q(\lambda)
		\tau_{(a,x)|\lambda}=\sum_{\lambda} q(\lambda)
		\sum_{k=1}^4 p(k|a,x,\lambda)E_{k|\lambda},
	\end{equation}
	where we have taken the simulation to be finite in $\lambda$. For the infinite case, an analogous argument holds. Let $v$ be a vector orthogonal to $\rho_{(a,x)}$, then
	\begin{equation}
		\expval{\rho_{(a,x)}}{v}=0\iff \sum_{\lambda} q(\lambda)
		\expval{ \tau_{(a,x)|\lambda}}{v}=0 \iff \expval{ \tau_{(a,x)|\lambda}}{v}=0 \quad \forall \lambda.
	\end{equation}
	Therefore,  by contrapositive if $v\in \text{supp}(\tau_{(a,x)|\lambda})$, then $v\in  \text{supp}(\rho_{(a,x)})$.  
	
	Let us fix the value $\lambda=\tilde \lambda$ and call $\{E_{k}\equiv E_{k|\tilde \lambda}\}_k$.  For any  pair $(a,x)$, since $\tr(\tau_{(a,x)|\tilde \lambda})=1$, there must exist a value $k_{(a,x)}$ for which $E_{k_{(a,x)}}\in \text{supp}(\tau_{(a,x)|\tilde \lambda})\implies E_{k_{(a,x)}}\in \text{supp}(\rho_{(a,x)})$. In particular, taking $(a,x)\in \{(1,1), (1,2)\}$, there must exist values $k_{(1,1)}, k_{(1,2)}$ such that
	\begin{equation}
		\left.
		\begin{aligned}
			E_{k_{(1,1)}}\in \text{span}\{C_1,C_2\}\\
			E_{k_{(1,2)}}\in \text{span}\{C_1,C_3\}
		\end{aligned}
		\right\}\implies 
		\left\{
		\begin{aligned}
			k_{(1,1)}=k_{(1,2)} &\implies E_{k_{(1,1)}}=E_{k_{(1,2)}}=C_1,\\
			k_{(1,1)}\neq k_{(1,2)} &\implies [E_{k_{(1,1)}}=C_2] \vee [E_{k_{(1,2)}}=C_3].
		\end{aligned}
		\right.
	\end{equation}
	This is, at least one vector of the $\tilde \lambda$ basis is a computational vector.  
	
	In the same fashion, if we take $(a,x)\in\{(3,1),(3,2)\}$, we conclude that at least one vector of the $\tilde \lambda$ basis is a Fourier vector. Nonetheless, $\abs{\braket{a}{f_b}}\neq0,1$ for all $a,b$, so they can never be  part of the same basis and we get into a contradiction. Thus, no classical simulation model for states exists for this set:
	\begin{equation}
		\text{Classical simulation of measurements } M_{a|x} \centernot \implies \text{Classical simulation of states } \rho_{(a,x)}.
	\end{equation}
	
	\subsection{Proof or implication 2}
	
	Let us give the counterexample now when $\tr(M_{a|x})\neq1$.  Consider the single SIC-POVM $\mathbf{M} = \lbrace M_a \rbrace_{a=1}^{4}$ in dimension $d=2$ \cite{Renes2004}. Let $\lbrace \ket{\psi_a}\rbrace_{a=1}^{4}$ be the fiducial vectors from which the SIC-POVM can be constructed, and define $\lbrace \Pi_{a} = \ketbra{\psi_a}{\psi_a} \rbrace_{a=1}^{4}$. Then,
	\begin{equation}\label{App_qubit_SIC}
		\mathbf{M} = \left\lbrace M_{a} = \frac{\Pi_{a}}{2}\right\rbrace_{a=1}^{4}, \quad \sum_{a=1}^{4} M_{a} = \openone, \quad \tr(M_{a}) = \frac{1}{2}.
	\end{equation}
	Now assume that each measurement is subject to isotropic noise, i.e. $v M_{a} + \frac{1-v}{d}\tr(M_{a})\openone$. For a single measurement setting, the simulation model in \eqref{App_Measurement_simulability} reduces to the one presented in \cite{Oszmaniec2017}, where it is also shown that a model always exists for the qubit noisy SIC-POVM if and only if $v \leq \sqrt{\frac{2}{3}} \coloneqq v_{\text{meas}}$.
	
	The equivalence with classical simulation of states would imply that a model also exists for the associated ensemble of states subject to isotropic noise, $\mathcal{E} = \lbrace v\Pi_{a} + \frac{1-v}{d}\openone \rbrace_{a=1}^{4}$. Since for qubits classical simulability of states is equivalent to joint measurability of the associated set of binary measurements $\lbrace \Pi_{a}, \openone - \Pi_{a} \rbrace_{a=1}^{4}$ \cite{Cobucci2025}, we can compute the critical visibility for which a simulation is possible through a semidefinite program \cite{Uola_2020}. This leads to a simulation for $v \leq \sqrt{\frac{1}{3}} \coloneqq v_{\text{states}}$.
	Therefore, since $v_{\text{meas}} > v_{\text{states}}$, 
	\begin{equation}
		\text{Classical simulation of measurements } M_{a} \centernot \implies \text{Classical simulation of states } \Pi_{a}.
	\end{equation}

	\subsection{Proof of implication 3}
	Consider that we have an ensemble of states $\mathcal{E} = \lbrace \rho_x \rbrace_{x=1}^{m}$ for which it is possible to find a classical simulation model \cite{Cobucci2025}, i.e.
	\begin{equation}\label{App_cl_sim_states}
		\rho_{x} = \int d\lambda \, q(\lambda) \, \tau_{x|\lambda}, \quad \text{where } [\tau_{x|\lambda},\tau_{x'|\lambda}] = 0.
	\end{equation}
	The commutation property implies that it is possible to decompose the states $\tau_{x|\lambda}$ as
	\begin{equation}
		\tau_{x|\lambda} = \sum_{k=1}^{d} p(k|x,\lambda) \ketbra{\phi_{k|\lambda}},
	\end{equation}
	where $\lbrace \ket{\phi_{k|\lambda}} \rbrace_{k}$ denotes a set of orthonormal bases for each $\lambda$. Therefore, the classical simulation model is
	\begin{equation}\label{App_classical_simulation_states}
		\rho_{x} = \int d\lambda \, q(\lambda) \, \sum_{k=1}^{d} p(k|x,\lambda) \ketbra{\phi_{k|\lambda}}.
	\end{equation}
	We want to prove that it is possible to find a simulation model for the associated binary measurement $\mathbf{M} = \lbrace M_{0|x} = \rho_{x}, M_{1|x} = \openone - \rho_x \rbrace_{x=1}^{m}$, i.e.
	\begin{equation}\label{App_Measurement_simulability}
		M_{a|x} = \int d\lambda \, q(\lambda) \, \sum_{k=1}^{d} p(a|x,k,\lambda) E_{k|\lambda},
	\end{equation}
	where $E_{k|\lambda} = \ketbra*{e_{k|\lambda}}e_{k|\lambda}$ are rank-1 projectors such that $E_{k|\lambda}E_{k'|\lambda} = \delta_{k,k'}E_{k|\lambda}$.

	From \cite{Cobucci2025}, we know that if a classical simulation model \eqref{App_classical_simulation_states} exists for the ensemble $\mathcal{E} = \lbrace \rho_{x} \rbrace_{x=1}^{m}$, then it also exists for the ensemble $\mathcal{E'} = \mathcal{E} \cup \lbrace \frac{\mathds{1} - \rho_x}{d-1} \rbrace_{x=1}^{m}$, where
	\begin{equation}
		\frac{\openone - \rho_x}{d-1} = \int d\lambda \, q(\lambda) \, \sum_{k=1}^{d} \frac{1-p(k|x,\lambda)}{d-1} \ketbra{\phi_{k|\lambda}} \quad \implies \quad \openone - \rho_x = \int d\lambda \, q(\lambda) \, \sum_{k=1}^{d} (1-p(k|x,\lambda)) \ketbra{\phi_{k|\lambda}}.
	\end{equation}
	Now we can define $p(0|x,k,\lambda) \coloneqq p(k|x,\lambda)$ and $p(1|x,k,\lambda) \coloneqq 1-p(k|x,\lambda)$. Since they satisfy $\sum_{a=0}^{1} p(a|x,k,\lambda) = 1$, they are a proper probability distribution. Therefore, we can write that
	\begin{equation}
		M_{0|x} = \rho_x = \int d\lambda \, q(\lambda) \, \sum_{k=1}^{d} p(0|x,k,\lambda) \ketbra{\phi_{k|\lambda}}, \quad M_{1|x} = \openone - \rho_{x} = \int d\lambda \, q(\lambda) \, \sum_{k=1}^{d} p(1|x,k,\lambda) \ketbra{\phi_{k|\lambda}},
	\end{equation}
	which is exactly the protocol for classical simulation of measurements \eqref{App_Measurement_simulability}.
	
	To see that the contrary is not true, we take the states $\{\rho_x\}_{x=1}^4$ in dimension $d=3$ given by:
	\begin{equation}
		\begin{aligned}
			\rho_1&=\frac{1}{2}(C_1+C_2),&
			\rho_2&=\frac{1}{2}(C_1+C_3),\\
			\rho_3&=\frac{1}{2}(F_1+F_2),&
			\rho_4&=\frac{1}{2}(F_1+F_3).
		\end{aligned}
	\end{equation}
	Notice that $\rho_1=M_{1|1},\rho_2=M_{1|2},\rho_3=M_{3|1},\rho_4=M_{3|2}$ from the previous subsection. Hence, the same exact reasoning shows they are classically simulable as measurements but not as states.

	\subsection{Proof of implication 4}
	Let us consider $d=2$.  Then, if $\tr(M_{a|x})=1$, since $\sum_a M_{a|x}=\openone$ we must have that the POVMs $\{M_{a|x}\}$ are binary: $\{M_{0|x},\openone-M_{0|x}\}_x$.  Notice then that $\rho_{(0, x)}\equiv M_{0|x}$ and $\rho_{(1, x)}\equiv\openone-M_{0|x}$ are indeed states.  Hence,  applying the result of the previous section, we have
	\begin{equation}
		\text{Classical simulation of states } \{\rho_{(a,x)}\}_{a,x} \implies \text{Classical simulation of measurements }  \{M_{0|x},\openone-M_{0|x} \}_x.
	\end{equation}
	
	For the other implication, we use that classical simulation of states for $d=2$ is equivalent to joint measurability of the binarization \cite{Cobucci2025}. In the main text we proved that if a set of measurement is classically simulable, then they are jointly measurable, i.e. classically simulable as states:
	\begin{equation}
		\begin{aligned}
			\text{Classical simulation of measurements } \{M_{0|x},\openone-M_{0|x}\}_x \implies& \text{Joint measurability of measurements }  \{M_{0|x},\openone-M_{0|x}\}_x \\
			\iff& \text{Classical simulation of states } \{\rho_{(a, x)}\}_{a,x}.
		\end{aligned}
	\end{equation}

	\section{Analytical witness for qubit measurements}\label{App_qubit_proof}
	In this section we prove an analytical solution to the maximisation problem
	\begin{equation}\label{App_beta_witness}
		\max_{\lbrace E_{k} \rbrace} \sum_{k=1}^{d} \tr(\sum_{a,x}D(a|x,k) \mathcal{O}_{ax} E_{k}),
	\end{equation}
	over qubit rank-one projective measurements $\{E_1,E_2\}$, where $D(a|x,k)$ is a specific deterministic strategy mapping $(x,k)$ to $a$.
	First, let us define the operators
	\begin{equation}
		\mathcal{O}_{\pm}=\frac{1}{2}\sum_{a,x}\left(D(a|x,1)\pm D(a|x,2)\right)\mathcal{O}_{ax}.
	\end{equation}
	It is straightforward to check that 
	\begin{equation}
		\begin{aligned}
			\sum_{a,x} D(a|x,\,1)\mathcal{O}_{a\,x}&=\mathcal{O}_{+}+\mathcal{O}_{-},\\
			\sum_{a,x} D(a|x,\,2)\mathcal{O}_{a\,x}&=\mathcal{O}_{+}-\mathcal{O}_{-}.
		\end{aligned}
	\end{equation}
	Replacing both relations in \eqref{App_beta_witness} leads to the following optimization problem:
	\begin{equation}
		\max_{\lbrace E_1,E_2 \rbrace}\left\{ \tr\left(\mathcal{O}_{+}\left(E_1+E_2\right)\right)+ \tr\left(\mathcal{O}_{-}\left(E_1-E_2\right)\right)\right\}=\tr\left(\mathcal{O}_{+}\right)+\max_{\lbrace E_1,E_2 \rbrace}\left\{\tr\left(\mathcal{O}_{-}\left(E_1-E_2\right)\right)\right\},
	\end{equation}
	where in the first step we have used that $E_1+E_2=\openone$. The final optimization can be performed by taking $E_1,E_2$ as the eigenvectors associated with the largest and smallest eigenvalues of $\mathcal{O}_{-}$, $\lambda_1$ and $\lambda_2$, respectively. Indeed,  for $E_k=\ketbra*{e_k}$:
	\begin{equation}
		\tr\left(\mathcal{O}_{-}\left(E_1-E_2\right)\right)= \ev{\mathcal{O}_{-}}{e_1}- \ev{\mathcal{O}_{-}}{e_2}\leq \lambda_1\left(\mathcal{O}_{-}\right)-\lambda_2\left(\mathcal{O}_{-}\right),
	\end{equation}
	where we have used that $\lambda_2 \left(\mathcal{O}_{-}\right)\leq  \ev{\mathcal{O}_{-}}{\psi} \leq \lambda_1\left(\mathcal{O}_{-}\right)$ for any state $\ket{\psi}$.
	Therefore,
	\begin{equation}
		\max_{\lbrace E_1,E_2 \rbrace}\left\{\tr\left(\mathcal{O}_{-}\left(E_1-E_2\right)\right)\right\}=\lambda_1\left(\mathcal{O}_{-}\right)-\lambda_2\left(\mathcal{O}_{-}\right).
	\end{equation}
	Notice that this expression gives the maximum independently of the signs of  $\lambda_i\left(\mathcal{O}_{-}\right)$. In addition, since we are dealing with 2-dimensional matrices, this expression can be related to matrix invariant, namely the trace and the determinant. Indeed,  it holds that
	\begin{equation}
		\lambda_1\left(\mathcal{O}_{-}\right)-\lambda_2\left(\mathcal{O}_{-}\right)=\sqrt{\left(\lambda_1\left(\mathcal{O}_{-}\right)+\lambda_2\left(\mathcal{O}_{-}\right)\right)^2-4\lambda_1\left(\mathcal{O}_{-}\right)\lambda_2\left(\mathcal{O}_{-}\right)}=\sqrt{\tr(\mathcal{O}_{-})^2-4\det(\mathcal{O}_{-})}.
	\end{equation}
	In consequence,  now the value of \eqref{App_beta_witness} reads
	\begin{equation}
		\max_{\lbrace E_{k} \rbrace} \sum_{k=1}^{d} \tr(\sum_{a,x}D(a|x,k) \mathcal{O}_{ax} E_{k}) = \tr(\mathcal{O}_{+})+\sqrt{\tr(\mathcal{O}_{-})^2-4\det(\mathcal{O}_{-})},
	\end{equation}
	which completes the proof.

	\section{Average non-disturbance does not imply classicality}\label{App_lower_dim_CS}
	
	In this Appendix we prove that average non-disturbance does not imply classicality for sets of POVMs in dimension $d \geq 5$.  For this purpose, we assume that average non-disturbance implies classicality and we show that this leads to a contradiction.
	
	Consider a $d$-dimensional POVM $\{\tilde M_a\}_{a=1}^{o}$ that, if performed twice in sequence, disturbs itself on average (later we will take this POVM to be the trine measurement \cite{Peres1990} to provide a specific case of the theorem). From the implication $\text{Classical model} \Rightarrow \text{Average non-disturbance}$, it is not classically simulable. Now construct the extended POVM
	\begin{equation}\label{App_CounterEx}
		\mathbf{M}\equiv \lbrace M_a=\ketbra{a} \oplus \tilde M_a \rbrace_{a}, \quad a = 1,\ldots,o.
	\end{equation}
	This measurement acts on a ($o$+$d$)-dimensional Hilbert space and is non-disturbing (and therefore average non-disturbing), since we can consider the instrument $\mathcal{I}_a^\dagger (\mathcal{X})=\tr(\mathcal{X}\ketbra{a}\oplus 0_d)M_a$.  Thus, assuming the antithesis, we must have that $\{M_a\}_{a=1}^{o}$ is classically simulable. As we will prove later, if an extended measurement $\lbrace M_a\rbrace_{a=1}^{o}$ of the form \eqref{App_CounterEx} is classically simulable, then the initial POVM $\lbrace \tilde{M}_a \rbrace_{a=1}^{o}$ is also classically simulable. Therefore, using again that $\text{Classical model} \Rightarrow \text{Average non-disturbance}$, we conclude that $\{\tilde M_a\}_{a=1}^{o}$ is average non-disturbing when performed twice in sequence. This last statement is in contradiction with the initial hypothesis, i.e. that  $\{\tilde M_a\}_{a=1}^{o}$ disturbs itself on average. Therefore, we get
	\begin{equation}
		\text{Average non-disturbance} \quad \not\Rightarrow \quad\text{Classical model}
	\end{equation}
	as long as a self-average-disturbing measurement $\{\tilde M_a\}_{a=1}^{o}$ exists for some dimension. One example is the case of the extremal trine qubit measurement $\mathbf{M}^{\text{trine}}\equiv \lbrace M_a = \frac{2}{3}\ketbra{\psi_a}\rbrace_a$, where $\ket{\psi_a} = \cos(a\pi/3)\ket{0} + \sin(a\pi/3)\ket{1}$ and $a=\lbrace 1,2,3 \rbrace$ \cite{Peres1990}. Since $[M_a,M_b] \neq 0$ for $a\neq b$, the trine measurement disturbs itself (indeed, for extremal POVMs on Bob's side, our problem of non-disturbance reduces to the one in \cite{Heinosaari2010}). Therefore, using this POVM in \eqref{App_CounterEx}, we get an example of an average non-disturbing measurement in dimension $5$ that is not classically simulable.

	Now we will show that $\{\tilde M_a\}_{a=1}^{o}$ is classically simulable if and only if $\{M_a\}_{a=1}^{o}$ is also classically simulable. For the sufficient condition, we assume that $\lbrace \tilde{M}_{a} \rbrace_{a=1}^{o}$ is classically simulable, i.e.
	\begin{equation}
		\tilde{M}_{a} = \int d\lambda \, q(\lambda) \, \tilde{F}_{a|\lambda},
	\end{equation}
	for some probability distribution $q(\lambda)$ and a set of projectors $\lbrace \tilde{F}_{a|\lambda}\rbrace_{a,\lambda}$ for which $\tilde{F}_{a|\lambda}\tilde{F}_{a'|\lambda} = \tilde{F}_{a|\lambda} \delta_{a,a'}$ for each $\lambda$. Then, a classical simulation protocol for the measurement $\{M_a=\ketbra{a} \oplus \tilde M_a\}_{a=1}^{o}$ can be built using a new set of projectors $\lbrace F_{a|\lambda} \equiv \ketbra{a} \oplus \tilde{F}_{a|\lambda} \rbrace_{a,\lambda}$. Indeed,
	\begin{equation}
		F_{a|\lambda}F_{a'|\lambda} = (\ketbra{a} \oplus \tilde{F}_{a|\lambda})(\ketbra{a'} \oplus \tilde{F}_{a'|\lambda}) = \delta_{a,a'} \ketbra{a} \oplus \delta_{a,a'} \tilde{F}_{a|\lambda} = \delta_{a,a'} F_{a|\lambda}, \quad \text{for each } \lambda.
	\end{equation}
	
	For the necessary condition, we assume that a classical simulation model exists for $\{M_a=\ketbra{a} \oplus \tilde M_a\}_{a=1}^{o}$, i.e.
	\begin{equation}\label{App_simulation_a_M}
		\ketbra{a} \oplus \tilde M_a = \int d\lambda \, q(\lambda) \, F_{a|\lambda},
	\end{equation}
	where now the projectors $\lbrace F_{a|\lambda} \rbrace_{a,\lambda}$ live in the $(o+d)$-dimensional Hilbert space $\mathcal{H}_{(o+d)}$. Denote as $\lbrace \ketbra{a} \oplus 0\rbrace_{a}$ a set of projectors acting non-trivially only on the $o$-dimensional subspace of $\mathcal{H}_{(o+d)}$ and such that $\lbrace \ket{a} \rbrace_{a}$ is a basis for this subspace. Then,
	\begin{equation}
		\delta_{a,a'}\delta_{a,a''} = (\bra{a'} \oplus 0_d) \ketbra{a} \oplus \tilde M_a (\ket{a''} \oplus 0_d) = \int d\lambda \, q(\lambda) \, (\bra{a'} \oplus 0_d) F_{a|\lambda} (\ket{a''} \oplus 0_d), \quad \forall a,a',a'',
	\end{equation}
	from which we deduce that, since $\lbrace F_{a|\lambda} \rbrace_a$ are projectors, the non-zero matrix elements on the $o$-dimensional subspace of $\mathcal{H}_{(o+d)}$ are diagonal, i.e.
	\begin{equation}
		F_{a|\lambda} = \ketbra{a} \oplus \tilde{F}_{a|\lambda}, \quad \forall a.
	\end{equation}
	Since $F_{a|\lambda}F_{a'|\lambda} = F_{a|\lambda} \delta_{a,a'}$, then also $\tilde{F}_{a|\lambda}\tilde{F}_{a'|\lambda} = \tilde{F}_{a|\lambda} \delta_{a,a'}$. Therefore, applying in \eqref{App_simulation_a_M} the projector $\tilde{\Pi}$ onto the $d$-dimensional subspace associated with $\lbrace \tilde{M}_{a} \rbrace_{a}$, we get
	\begin{equation}
		\tilde{M}_{a} = \tilde{\Pi}M_{a} \tilde{\Pi} = \int d\lambda \, q(\lambda) \tilde{\Pi} F_{a|\lambda} \tilde{\Pi} = \int d\lambda \, q(\lambda) \tilde{F}_{a,\lambda},
	\end{equation}
	which is a classical simulation model for $\lbrace \tilde{M}_{a} \rbrace_{a}$.
	

\end{document}